\documentclass[twocolumn,10pt]{IEEEtran}

\usepackage{amsmath,amssymb,amsthm}
\usepackage{graphicx}
\usepackage[acronym]{glossaries}
\usepackage{xcolor}
\usepackage{comment}
\usepackage{cases}
\usepackage{algorithm}
\usepackage{algcompatible}
\usepackage{setspace}
\usepackage{booktabs}
\usepackage{subfigure}

\newacronym{6g}{6G}{sixth-generation}
\newacronym{5g}{5G}{fifth-generation}
\newacronym{mimo}{MIMO}{multiple-input multiple-output}
\newacronym{rf}{RF}{radio frequency}
\newacronym{adc}{ADC}{analog-to-digital converter}
\newacronym{dac}{DAC}{digital-to-analog converter}
\newacronym{svd}{SVD}{singular value decomposition}
\newacronym{irs}{IRS}{intelligent reflecting surface}
\newacronym{ris}{RIS}{reconfigurable intelligent surface}
\newacronym{bd-ris}{BD-RIS}{beyond diagonal RIS}
\newacronym{sim}{SIM}{stacked intelligent metasurface}
\newacronym{milac}{MiLAC}{microwave linear analog computer}
\newacronym{awgn}{AWGN}{additive white Gaussian noise}
\newacronym{iid}{i.i.d.}{independent and identically distributed}
\newacronym{snr}{SNR}{signal-to-noise ratio}
\newacronym{em}{EM}{electromagnetic}
\newacronym{dma}{DMA}{dynamic metasurface antenna}
\newacronym{bd-dma}{BD-DMA}{beyond diagonal DMA}
\newacronym{zfbf}{ZFBF}{zero-forcing beamforming}
\newacronym{mmse}{MMSE}{minimum mean square error}
\newacronym{qpsk}{QPSK}{quadrature phase shift keying}
\newacronym{dsa}{DSA}{dynamic scattering array}

\newtheorem{definition}{Definition}
\newtheorem{proposition}{Proposition}
\newtheorem{corollary}{Corollary}
\newtheorem{remark}{Remark}

\definecolor{myred}{RGB}{184, 84, 80}
\definecolor{mygreen}{RGB}{130, 179, 102}
\definecolor{myyellow}{RGB}{214, 182, 86}
\definecolor{myblue}{RGB}{108, 142, 191}

\begin{document}
\bstctlcite{BSTcontrol}

\title{MIMO Systems Aided by Microwave Linear Analog Computers: Capacity-Achieving Architectures with Reduced Circuit Complexity}

\author{Matteo~Nerini,~\IEEEmembership{Member,~IEEE}, and
        Bruno~Clerckx,~\IEEEmembership{Fellow,~IEEE}

\thanks{This work has been supported in part by UKRI under Grant EP/Y004086/1, EP/X040569/1, EP/Y037197/1, EP/X04047X/1, EP/Y037243/1.
Corresponding author: Bruno Clerckx.}
\thanks{Matteo Nerini and Bruno Clerckx are with the Department of Electrical and Electronic Engineering, Imperial College London, SW7 2AZ London, U.K. (e-mail: m.nerini20@imperial.ac.uk; b.clerckx@imperial.ac.uk).}
\thanks{Bruno Clerckx is also with the Department of Electronic Engineering, Kyung Hee University, Yongin-si, Gyeonggi-do 17104, South Korea.}}

\maketitle

\begin{abstract}
To meet the demands of future wireless networks, antenna arrays must scale from massive \gls{mimo} to gigantic \gls{mimo}, involving even larger numbers of antennas.
To address the hardware and computational cost of gigantic \gls{mimo}, several strategies are available that shift processing from the digital to the analog domain.
Among them, \glspl{milac} offer a compelling solution by enabling fully analog beamforming through reconfigurable microwave networks.
Prior work has focused on fully-connected \glspl{milac}, whose ports are all interconnected to each other via tunable impedance components.
Although such \glspl{milac} are capacity-achieving, their circuit complexity, given by the number of required impedance components, scales quadratically with the number of antennas, limiting their practicality.
To solve this issue, in this paper, we propose a graph theoretical model of \gls{milac} facilitating the systematic design of lower-complexity \gls{milac} architectures.
Leveraging this model, we propose stem-connected \glspl{milac} as a family of \gls{milac} architectures maintaining capacity-achieving performance while drastically reducing the circuit complexity.
Besides, we optimize stem-connected \glspl{milac} with a closed-form capacity-achieving solution.
Our theoretical analysis, confirmed by numerical simulations, shows that stem-connected \glspl{milac} are capacity-achieving, but with circuit complexity that scales linearly with the number of antennas, enabling high-performance, scalable, gigantic \gls{mimo}.
\end{abstract}

\glsresetall

\begin{IEEEkeywords}
Capacity, gigantic MIMO, graph theory, microwave linear analog computer (MiLAC)
\end{IEEEkeywords}

\section{Introduction}

As wireless systems evolve toward \gls{6g}, the demand for ultra-high data rates, low latency, and pervasive connectivity continues to increase, pushing beyond the capabilities of \gls{5g} technologies.
While massive \gls{mimo} systems, typically including tens of antennas, have proven successful in \gls{5g} \cite{boc14,lar14}, \gls{6g} requires a further substantial increase in the number of antennas.
This emerging paradigm, denoted as gigantic \gls{mimo}, envisions \gls{mimo} architectures with antenna arrays scaling to the order of thousands of antennas \cite{qua22,bjo24}.
Such extreme scaling offers significant enhancements in spatial multiplexing and beamforming resolution, which are critical to meet \gls{6g} targets in the anticipated upper mid-band spectrum (7-24~GHz).
However, traditional digital \gls{mimo} architectures are not suited to this scale due to the prohibitive cost, power consumption, and complexity of assigning one \gls{rf} chain per antenna, each requiring high-resolution \glspl{adc}/\glspl{dac} and mixers.
This challenge has steered research efforts toward analog-domain solutions, which can provide more scalable, cost-effective, and energy-efficient alternatives to fully digital \gls{mimo} architectures.

Hybrid digital-analog \gls{mimo} architectures originally emerged to reduce the demand of \gls{rf} chains while preserving high performance through the use of analog phase shifters \cite{aya14,soh16}.
Another technology classically used to change the radiation pattern of an antenna array operating in the analog domain is reconfigurable antennas \cite{hau13}.
The idea of reconfigurable antennas traces back to the 1930s, when they were realized through mechanically moving parts.
Since the 1990s, reconfigurable antennas have been realized by interconnecting parts of the antenna through tunable components, such as varactors, mechanical \gls{rf} switches, or semiconductor \gls{rf} switches.
To synergize the benefits of hybrid beamforming and the recent advances in reconfigurable antennas, the tri-hybrid \gls{mimo} architecture has been recently proposed \cite{cas25}, where the signal is processed in three stages: digitally, in the analog domain by phase shifts, and in the \gls{em} domain by reconfigurable antennas.

Beamforming in the analog domain can also be achieved through the use of \glspl{irs}, also known as \glspl{ris}.
A \gls{ris} is a surface made of a multitude of elements with reconfigurable scattering capabilities, which can be integrated at the transceiver device to reconfigure its effective radiation pattern \cite{wu21,hua23}.
To enhance the flexibility of \gls{ris}, the concept of \gls{bd-ris} has emerged, where the \gls{ris} elements are interconnected to each other through tunable interconnections \cite{she22,li24,mis24}.
Furthermore, the use of multiple stacked \glspl{ris} has been explored to increase adaptability, giving rise to \gls{sim} technology \cite{an23,an25,liu25a,liu25b}.
While \cite{wu21}-\cite{liu25b} have considered reconfigurable surfaces to manipulate the incident \gls{em} signal, \glspl{dma} and \gls{dsa} have emerged to manipulate how the \gls{em} signal is radiated and received \cite{shl21,dar24}.
\Glspl{bd-dma} have been recently proposed as an advanced version of \glspl{dma}, where tunable interconnections between the meta-atoms enable reconfigurable coupling and hence increase the performance \cite{pro25}.

While analog-domain solutions are particularly appealing for gigantic \gls{mimo} beamforming, they commonly suffer from two limitations.
First, they typically still rely on digital pre-processing at the transmitter and digital post-processing at the receiver, which could require large hardware and computational complexity in gigantic \gls{mimo}.
Second, beamforming in the analog domain generally suffers from limited reconfigurability and cannot achieve the same performance as fully digital beamforming.
To address both limitations, \gls{milac} has recently emerged as a promising approach for gigantic \gls{mimo} systems \cite{ner25-1,ner25-2}.
A \gls{milac} is a multiport microwave network made of tunable impedance components, having input and output ports.
At the input ports, a signal is applied, which is processed by the \gls{milac} in the analog domain as it propagates within the microwave network.
Thus, at the output ports, the processed signal can be read.
Exploiting these analog processing capabilities, we can use a \gls{milac} at the transmitter side to precode the transmitted symbols by feeding them to the \gls{milac} input ports through \gls{rf} chains, and connecting the \gls{milac} output ports to the transmitting antennas.
Similarly, at the receiver side, we can use a \gls{milac} to combine the received signal by connecting the receiving antennas to the \gls{milac} input ports and sampling the signal at the \gls{milac} output ports via \gls{rf} chains.

By precoding and combining the signal fully in the analog domain, \gls{milac} offers five distinct advantages for gigantic \gls{mimo} in terms of hardware and computational complexity \cite{ner25-2}.
\textit{First}, it offers the same flexibility as digital beamforming, and hence maximum performance.
\textit{Second}, it only requires as many \gls{rf} chains as the transmitted symbols, i.e., data streams.
\textit{Third}, it can operate with low-resolution \glspl{adc} and \glspl{dac}, since the signals carried and sampled at the \gls{rf} chains are the symbols up to a scaling factor, which lie in a constellation with limited cardinality.
For example, in the case of \gls{qpsk}, the signals can be chosen from the constellation $\{+1+j,+1-j,-1-j,-1+j\}$, up to a scaling factor, which can be generated through \glspl{dac} having just 1-bit resolution for both in-phase and quadrature parts of the signal.
\textit{Fourth}, it does not require any computation at each symbol time, as the symbols are precoded and combined fully in the analog domain.
Consequently, the symbol vector does not need to be multiplied by a beamforming matrix, unlike in digital or hybrid beamforming.
\textit{Fifth}, at each channel coherence time, it can reconfigure the beamforming matrix for \gls{zfbf} with significantly reduced computational complexity, as shown in \cite{ner25-2}.
In detail, \gls{zfbf} and \gls{mmse} combining can be executed with a computational complexity growing quadratically with the number of antennas, rather than cubically as in digital \gls{mimo} \cite{ner25-2}.

To analyze the fundamental performance limits of \gls{milac}, in \cite{ner25-1} and \cite{ner25-2} it has been assumed that its tunable impedance components can be arbitrarily reconfigured.
A more practical design has been considered in \cite{ner25-3}, where a \gls{milac} made of lossless and reciprocal tunable impedance components has been considered, implementable with varactor diodes or PIN diodes \cite{hau13}.
Interestingly, even under the lossless and reciprocal constraints, it has been shown that \gls{milac}-aided beamforming can achieve the same capacity as fully digital beamforming.
However, the number of tunable components included in the \gls{milac} architecture considered in previous work \cite{ner25-1,ner25-2,ner25-3} grows quadratically with the number of antennas and streams, which could become prohibitive for very large numbers of antennas and/or streams.
To address this issue, in this paper, we propose a novel family of \gls{milac} architectures with a significantly reduced circuit complexity, which are proven to remain capacity-achieving when made of lossless and reciprocal components.
Specifically, the contribution of this paper can be summarized as follows.

\textit{First}, we propose a model of \gls{milac} based on graph theory in Section~\ref{sec:graph}.
With this model, we represent any \gls{milac} architecture through a graph, whose vertices correspond to the \gls{milac} ports and edges correspond to the tunable admittance components interconnecting them.
While prior work has focused on fully-connected \glspl{milac}, where every port is interconnected to all other ports, our graph theoretical modeling enables the exploration of different \gls{milac} architectures, with sparse interconnections and therefore reduced circuit complexity.

\textit{Second}, we characterize a family of \gls{milac} architectures, named stem-connected \glspl{milac}, which is proven to achieve capacity in point-to-point \gls{mimo} systems in Section~\ref{sec:architectures}.
Compared to a fully-connected \gls{milac}, a stem-connected \gls{milac} has a significantly reduced circuit complexity, measured by the number of tunable admittance components interconnecting its ports.
While the circuit complexity of a fully-connected \gls{milac} is $\mathcal{O}((N_S+N_T)^2)$, where $N_S$ is the number of transmitted streams and $N_T$ is the number of transmitting antennas, the circuit complexity of a stem-connected \gls{milac} is $\mathcal{O}(N_SN_T)$.

\textit{Third}, we optimize stem-connected \glspl{milac} at both transmitter and receiver sides in Sections~\ref{sec:tx} and \ref{sec:rx}, respectively, and demonstrate that they can achieve the capacity for any channel realization.
To this end, we first formulate a feasibility check problem whose solution, if it exists, is proven to satisfy a sufficient and necessary condition for achieving capacity.
We then solve the obtained problem in closed form, proving that stem-connected \glspl{milac} are capacity-achieving.

\textit{Fourth}, we provide numerical results in Section~\ref{sec:results} to support our theoretical findings.
Monte Carlo simulations confirm that a stem-connected \gls{milac} achieves channel capacity for any channel realization, as a fully-connected \gls{milac}.
At the same time, the stem-connected \gls{milac} architecture offers significantly lower circuit complexity.
While the number of tunable admittance components in a fully-connected \gls{milac} scales quadratically with the number of antennas, it scales only linearly in a stem-connected \gls{milac}.


\textit{Notation}:
Vectors and matrices are denoted with bold lower and bold upper letters, respectively.
Scalars are represented with letters not in bold font.
$\Re\{a\}$, $\Im\{a\}$, and $\vert a\vert$ refer to the real part, imaginary part, and absolute value of a complex scalar $a$, respectively.
$\mathbf{a}^*$, $\mathbf{a}^T$, $\mathbf{a}^H$, $[\mathbf{a}]_{i}$, and $\Vert\mathbf{a}\Vert$ refer to the conjugate, transpose, conjugate transpose, $i$th element, and $l_{2}$-norm of a vector $\mathbf{a}$, respectively.
$\mathbf{A}^*$, $\mathbf{A}^T$, $\mathbf{A}^H$, $[\mathbf{A}]_{i,k}$, $[\mathbf{A}]_{i,:}$, and $[\mathbf{A}]_{:,k}$ refer to the conjugate, transpose, conjugate transpose, $(i,k)$th element, $i$th row, and $k$th column of a matrix $\mathbf{A}$, respectively.
$[\mathbf{A}]_{\mathcal{I},\mathcal{K}}$ refers to the submatrix of $\mathbf{A}$ obtained by selecting the rows and columns indexed by the elements of the sets $\mathcal{I}$ and $\mathcal{K}$, respectively.
$\mathbb{R}$ and $\mathbb{C}$ denote the real and complex number sets, respectively.
$j=\sqrt{-1}$ denotes the imaginary unit.
$\mathbf{I}_N$ and $\mathbf{0}_N$ denote the identity matrix and the all-zero matrix with dimensions $N\times N$, respectively.
$\mathbf{0}_{M\times N}$ denotes the all-zero matrix with dimensions $M\times N$.

\section{System Model}
\label{sec:system}

\begin{figure}[t]
\centering
\includegraphics[width=0.48\textwidth]{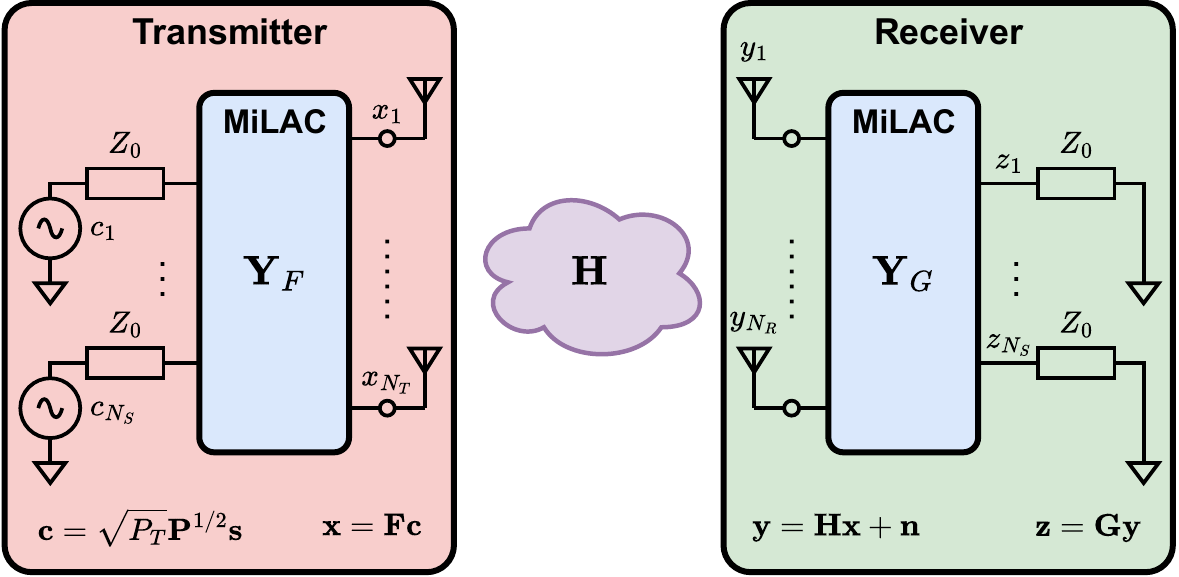}
\caption{MiLAC-aided MIMO system model.}
\label{fig:system}
\end{figure}

Consider a narrowband point-to-point \gls{mimo} system aided by a \gls{milac} at both the transmitter and receiver, as introduced in \cite[Section~II]{ner25-3} and represented in Fig.~\ref{fig:system}.
The $N_T$-antenna transmitter sends $N_S$ symbols (or streams) to the $N_R$-antenna receiver, where $N_S\leq\min\{N_T,N_R\}$.
Following the model in \cite[Section~II]{ner25-3}, the signal used for detection at the receiver $\mathbf{z}\in\mathbb{C}^{N_S\times 1}$ is
\begin{equation}
\mathbf{z}=\sqrt{P_T}\mathbf{G}\mathbf{H}\mathbf{F}\mathbf{P}^{1/2}\mathbf{s}+\mathbf{G}\mathbf{n},\label{eq:system-model}
\end{equation}
where $\mathbf{s}\in\mathbb{C}^{N_S\times1}$ is the symbol vector such that $\mathbb{E}[\mathbf{s}\mathbf{s}^H]=\mathbf{I}_{N_S}$, $\mathbf{P}^{1/2}\in\mathbb{C}^{N_S\times N_S}$ is the square root of the power allocation matrix given by $\mathbf{P}^{1/2}=\text{diag}(\sqrt{p_1},\ldots,\sqrt{p_{N_S}})$, with $p_{s}$ denoting the power allocation for the $s$th symbol such that $\sum_{s=1}^{N_S}p_{s}=1$, and $P_T$ is the transmitted power.
In addition, $\mathbf{F}\in\mathbb{C}^{N_T\times N_S}$ is the precoding matrix implemented by the transmitter-side \gls{milac}, $\mathbf{H}\in\mathbb{C}^{N_R\times N_T}$ is the wireless channel, and $\mathbf{G}\in\mathbb{C}^{N_S\times N_R}$ is the combining matrix implemented by the receiver-side \gls{milac}.
Besides, $\mathbf{n}\in\mathbb{C}^{N_R\times 1}$ is the \gls{awgn} such that $\mathbb{E}[\mathbf{n}\mathbf{n}^H]=\sigma^2\mathbf{I}_{N_R}$, with $\sigma^2$ denoting the noise power.

According to \cite{ner25-2}, the precoding matrix $\mathbf{F}$ is a function of the admittance matrix of the \gls{milac} at the transmitter $\mathbf{Y}_F\in\mathbb{C}^{(N_S+N_T)\times(N_S+N_T)}$ and writes as
\begin{equation}
\mathbf{F}=\left[\left(\frac{\mathbf{Y}_F}{Y_0}+\mathbf{I}_{N_S+N_T}\right)^{-1}\right]_{N_S+\mathcal{N}_T,\mathcal{N}_S},\label{eq:F(Y)}
\end{equation}
where $\mathcal{N}_T=\{1,\ldots,N_T\}$ and $\mathcal{N}_S=\{1,\ldots,N_S\}$\footnote{We recall that the product $\mathbf{F}\mathbf{c}$ of the precoding matrix $\mathbf{F}$ with the power-allocated symbol vector $\mathbf{c}=\sqrt{P_T}\mathbf{P}^{1/2}\mathbf{s}$ is performed in the analog domain by the \gls{milac}, which receives in input the signal $\mathbf{c}$ (see Fig.~\ref{fig:system}).}.
In words, the precoding matrix $\mathbf{F}$ is the $N_T\times N_S$ matrix obtained by selecting the last $N_T$ rows and the first $N_S$ colums of $(\mathbf{Y}_F/Y_0+\mathbf{I}_{N_S+N_T})^{-1}$.
Considering the \gls{milac} at the transmitter to be composed of a lossless and reciprocal microwave network, its admittance matrix $\mathbf{Y}_F$ is purely imaginary and symmetric according to microwave network theory \cite[Chapter~4]{poz11}, i.e., $\mathbf{Y}_F=j\mathbf{B}_F$, with $\mathbf{B}_F=\mathbf{B}_F^T$, where $\mathbf{B}_F\in\mathbb{R}^{(N_S+N_T)\times(N_S+N_T)}$ is the susceptance matrix of the \gls{milac}.
In alternative to the admittance matrix, the transmitter-side \gls{milac} can be represented with its scattering matrix 
$\boldsymbol{\Theta}_F\in\mathbb{C}^{(N_S+N_T)\times(N_S+N_T)}$, related to $\mathbf{Y}_F$ via
\begin{equation}
\boldsymbol{\Theta}_F=\left(Y_0\mathbf{I}_{N_S+N_T}+\mathbf{Y}_F\right)^{-1}\left(Y_0\mathbf{I}_{N_S+N_T}-\mathbf{Y}_F\right),\label{eq:TF(YF)}
\end{equation}
and the precoding matrix $\mathbf{F}$ in \eqref{eq:F(Y)} can be simplified as
\begin{equation}
\mathbf{F}=\frac{1}{2}\left[\boldsymbol{\Theta}_F\right]_{N_S+\mathcal{N}_T,\mathcal{N}_S},
\end{equation}
as a function of $\boldsymbol{\Theta}_F$, equivalently to \eqref{eq:F(Y)} \cite{ner25-3}.
In the case of a lossless and reciprocal \gls{milac}, its scattering matrix is unitary and symmetric, i.e., $\boldsymbol{\Theta}_F^H\boldsymbol{\Theta}_F=\mathbf{I}_{N_S+N_T}$, $\boldsymbol{\Theta}_F=\boldsymbol{\Theta}_F^{T}$ \cite{ner25-3}.

At the receiver side, the combining matrix $\mathbf{G}$ is given by
\begin{equation}
\mathbf{G}=\left[\left(\frac{\mathbf{Y}_G}{Y_0}+\mathbf{I}_{N_R+N_S}\right)^{-1}\right]_{N_R+\mathcal{N}_S,\mathcal{N}_R},\label{eq:G(Y)}
\end{equation}
where $\mathcal{N}_R=\{1,\ldots,N_R\}$, as a function of the admittance matrix of the \gls{milac} at the receiver $\mathbf{Y}_G\in\mathbb{C}^{(N_R+N_S)\times(N_R+N_S)}$.
For a lossless and reciprocal \gls{milac}, $\mathbf{Y}_G$ is subject to $\mathbf{Y}_G=j\mathbf{B}_G$, $\mathbf{B}_G=\mathbf{B}_G^T$, with $\mathbf{B}_G\in\mathbb{R}^{(N_R+N_S)\times(N_R+N_S)}$ being the susceptance matrix of the \gls{milac}.
Alternatively, we can introduce the scattering matrix of the \gls{milac} at the receiver $\boldsymbol{\Theta}_G\in\mathbb{C}^{(N_R+N_S)\times(N_R+N_S)}$ as
\begin{equation}
\boldsymbol{\Theta}_G=\left(Y_0\mathbf{I}_{N_R+N_S}+\mathbf{Y}_G\right)^{-1}\left(Y_0\mathbf{I}_{N_R+N_S}-\mathbf{Y}_G\right),\label{eq:TG(YG)}
\end{equation}
and express the combining matrix $\mathbf{G}$ as
\begin{equation}
\mathbf{G}=\frac{1}{2}\left[\boldsymbol{\Theta}_G\right]_{N_R+\mathcal{N}_S,\mathcal{N}_R}.
\end{equation}
equivalently to \eqref{eq:G(Y)} \cite{ner25-3}, where $\boldsymbol{\Theta}_G^H\boldsymbol{\Theta}_G=\mathbf{I}_{N_R+N_S}$ and $\boldsymbol{\Theta}_G=\boldsymbol{\Theta}_G^{T}$ for a lossless and reciprocal \gls{milac}.

Given the system model in \eqref{eq:system-model}, the achievable rate that can be obtained with \gls{milac}-aided beamforming is
\begin{equation}
R=\sum_{s=1}^{N_S}\log_2\left(1+\frac{P_Tp_{s}\left\vert[\mathbf{G}\mathbf{H}\mathbf{F}]_{s,s}\right\vert^2}
{P_T\sum_{t\neq s}p_{t}\left\vert[\mathbf{G}\mathbf{H}\mathbf{F}]_{s,t}\right\vert^2+\left\Vert[\mathbf{G}]_{s,:}\right\Vert^2\sigma^2}\right),\label{eq:C}
\end{equation}
where inter-stream interference is treated as noise since the receiver combines the symbols purely in the analog domain, avoiding further digital processing \cite{ner25-3}.
By optimally reconfiguring the admittance matrices of the \glspl{milac} $\mathbf{Y}_F$ and $\mathbf{Y}_G$ and the power allocations $p_1,\ldots,p_{N_S}$, the maximum achievable rate, i.e., the capacity, is given by
\begin{equation}
C=\sum_{s=1}^{N_S}\log_2\left(1+\frac{P_Tp_{s}^\star\lambda_{s}}{4\sigma^2}\right),\label{eq:Cstar}
\end{equation}
where $\lambda_s$ is the $s$th eigenvalue of $\mathbf{H}\mathbf{H}^H$ and $p_s^\star$ is the optimal power allocation for the $s$th symbol obtained through water-filling, as discussed in \cite{ner25-3}\footnote{Note that we refer to the maximum achievable rate in \eqref{eq:C} as the capacity with a slight abuse of terminology, since \eqref{eq:C} is strictly speaking the capacity only when $N_S=\min\{N_T,N_R\}$.}.

It has been shown in \cite{ner25-3} that the capacity in \eqref{eq:Cstar} can be achieved by lossless and reciprocal \glspl{milac}, having tunable admittance components in their multiport microwave networks interconnecting each port to each other.
However, the circuit complexity of such \glspl{milac}, defined as the number of tunable admittance components in their microwave networks, grows quadratically with the number of antennas, becoming prohibitive in gigantic \gls{mimo} systems.
To address this problem, we investigate whether there exist lower-complexity \gls{milac} architectures still able to achieve the capacity.

\section{Graph Theoretical Modeling of MiLAC}
\label{sec:graph}

In this section, we model \gls{milac} architectures as graphs, inspired by the graph theoretical modeling of \gls{bd-ris} introduced in \cite{ner24}.
This modeling shows that numerous \gls{milac} architectures are available to balance flexibility and circuit complexity.

A \gls{milac} has been introduced in \cite{ner25-1} as a multiport microwave network made of tunable admittance components interconnecting the ports to ground and to each other.
Considering a \gls{milac} with $N_V$ ports, we denote as $Y_{n,n}\in\mathbb{C}$ the admittance connecting port $n$ to ground, for $n=1,\ldots,N_V$, and as $Y_{m,n}\in\mathbb{C}$ the admittance interconnecting ports $m$ and $n$, $\forall m\neq n$.
As a function of these tunable admittance components, the admittance matrix of the \gls{milac} $\mathbf{Y}\in\mathbb{C}^{N_V\times N_V}$ writes as
\begin{equation}
\left[\mathbf{Y}\right]_{m,n}=
\begin{cases}
-Y_{m,n} & m\neq n\\
\sum_{k=1}^{N_V}Y_{k,n} & m=n
\end{cases},\label{eq:Yik}
\end{equation}
for $m,n=1,\ldots,N_V$ \cite{ner25-1}.
In a \gls{milac} where the ports are all interconnected to each other, i.e., including all admittance components $Y_{m,n}$, for $m,n=1,\ldots,N_V$, $\mathbf{Y}$ can be any arbitrary matrix following \eqref{eq:Yik}.
In the case of lossless and reciprocal admittance components, $\mathbf{Y}$ can be any arbitrary imaginary symmetric matrix, as considered in Section~\ref{sec:system}.
Nevertheless, besides this \gls{milac} architecture where the ports are all interconnected to each other, it is also possible to design \gls{milac} architectures including only a reduced number of interconnections between ports, and hence having lower circuit complexity.
If ports $m$ and $n$ are not interconnected, $Y_{m,n}$ is an open circuit, i.e., it is forced to zero.
In practice, two ports can be coupled even if disconnected due to mutual coupling effects, and therefore the distances between ports should be at least half a wavelength to mitigate these effects and a careful layout design should minimize unintended coupling as in standard microwave circuits.
As a consequence of $Y_{m,n}=0$, also the $(m,n)$th element of $\mathbf{Y}$ is forced to zero according to \eqref{eq:Yik}, limiting the flexibility of the resulting \gls{milac}.
This observation enables us to design numerous \gls{milac} architectures that strike a balance between flexibility and circuit complexity.

To rigorously characterize all possible \gls{milac} architectures, we represent the circuit topology of any $N_V$-port \gls{milac} through a graph, i.e., a set of ``dots'' (named vertices) interconnected by ``lines'' (named edges).
Formally, such graph is simple and undirected, meaning that the edges do not have a direction and that any two vertices are connected by at most one edge, and is defined by a pair of sets $\mathcal{G}=(\mathcal{V},\mathcal{E})$, where $\mathcal{V}$ and $\mathcal{E}$ are its vertex set and edge set, respectively \cite{bon76}.
The set $\mathcal{V}$ is constituted by the indexes of the ports of the \gls{milac}, i.e.,
\begin{equation}
\mathcal{V}=\left\{V_1,\ldots,V_{N_V}\right\},
\end{equation}
where $N_V$ denotes the number of \gls{milac} ports, which is $N_V=N_S+N_T$ or $N_V=N_R+N_S$ for a transmitter- or receiver-side \gls{milac}, respectively.
The set $\mathcal{E}$ is defined as
\begin{equation}
\mathcal{E}=\left\{\left(V_n,V_m\right)\mid V_n,V_m\in\mathcal{V},\;Y_{n,m}\neq0,\;n\neq m\right\},
\end{equation}
meaning that there is an edge between vertices $V_m$ and $V_n$ if and only if there is a tunable admittance component interconnecting ports $m$ and $n$.
In a lossless and reciprocal \gls{milac} with associated graph $\mathcal{G}=(\mathcal{V},\mathcal{E})$, the admittance matrix is therefore subject to $\mathbf{Y}=j\mathbf{B}$, with the susceptance matrix $\mathbf{B}$ fulfilling $\mathbf{B}=\mathbf{B}^T$ and $\mathbf{B}\in\mathcal{B}_{\mathcal{G}}$, where
\begin{multline}
\mathcal{B}_{\mathcal{G}}=\left\{\mathbf{B}\in\mathbb{R}^{N_V\times N_V}\mid\left[\mathbf{B}\right]_{n,m}=0,\;\forall n\neq m,\right.\\
\left(V_n,V_m\right)\notin\mathcal{E}\Bigl\},\label{eq:B-set}
\end{multline}
indicating that the $(m,n)$th element of $\mathbf{Y}$ is forced to zero if ports $m$ and $n$ are not interconnected by a tunable admittance.
Besides, the circuit complexity of such a \gls{milac}, defined as the number of tunable admittance components, is given by $N_C=N_V+N_E$, where $N_E$ denotes the cardinality of $\mathcal{E}$, since there are $N_V$ tunable admittance components connecting each port to ground and $N_E$ components interconnecting the ports.

The \gls{milac} architecture considered in previous works \cite{ner25-1,ner25-2,ner25-3}, that we denote as fully-connected \gls{milac}, has an associated graph $\mathcal{G}$ that is the complete graph on $N_V$ vertices.
Therefore, in this case we have $\mathcal{B}_{\mathcal{G}}=\mathbb{R}^{N_V\times N_V}$, yielding maximum flexibility, and $N_E=N_V(N_V-1)/2$ edges in $\mathcal{G}$, yielding a circuit complexity of
\begin{equation}
N_C^\text{Fully}=\frac{N_V\left(N_V+1\right)}{2}.\label{eq:complete-complexity}
\end{equation}
For a transmitter-side \gls{milac}, where $N_V=N_S+N_T$, we have a circuit complexity of
\begin{equation}
N_C^\text{Fully}=\frac{\left(N_S+N_T\right)\left(N_S+N_T+1\right)}{2},\label{eq:fully-tx-complexity}
\end{equation}
growing with $\mathcal{O}((N_S+N_T)^2)$, i.e., quadratically with the number of antennas $N_T$.
A similar discussion holds for a \gls{milac} at the receiver and, by replacing $N_T$ with $N_R$ in \eqref{eq:fully-tx-complexity}, we obtain that its complexity scales with $\mathcal{O}((N_S+N_R)^2)$, i.e., quadratically with $N_R$.
Although the fully-connected \gls{milac} offers the highest flexibility, corresponding to the largest set $\mathcal{B}_{\mathcal{G}}$, it also incurs the highest circuit complexity, which scales quadratically with the number of antennas.
This motivates the design of alternative \gls{milac} architectures with reduced circuit complexity that ideally do not significantly compromise the achievable rate despite their reduced flexibility.

\section{Capacity-Achieving MiLAC Architectures}
\label{sec:architectures}

We have introduced a graph theoretical model for \gls{milac} and showed that numerous \gls{milac} architectures are available to balance flexibility and circuit complexity.
In this section, we propose a family of \gls{milac} architectures proven to be capacity-achieving despite their highly reduced circuit complexity.

\subsection{Definition of Capacity-Achieving MiLAC Architectures}

We begin by recalling under which conditions we can achieve capacity in a \gls{milac}-aided system, and providing a formal definition of capacity-achieving \gls{milac} architecture.
According to \cite{ner25-3}, the capacity in \eqref{eq:Cstar} can be achieved by setting the scattering matrix of the transmitter-side \gls{milac} $\boldsymbol{\Theta}_F$ such that $[\boldsymbol{\Theta}_F]_{N_S+\mathcal{N}_T,\mathcal{N}_S}=\bar{\mathbf{V}}$, where $\bar{\mathbf{V}}\in\mathbb{C}^{N_T\times N_S}$ contains the first $N_S$ right singular vectors of $\mathbf{H}$, and the scattering matrix of the receiver-side \gls{milac} $\boldsymbol{\Theta}_G$ such that $[\boldsymbol{\Theta}_G]_{N_R+\mathcal{N}_S,\mathcal{N}_R}=\bar{\mathbf{U}}^H$, where $\bar{\mathbf{U}}\in\mathbb{C}^{N_R\times N_S}$ contains the first $N_S$ left singular vectors of $\mathbf{H}$.

Following these capacity-achieving conditions, a \gls{milac} at the transmitter is defined as capacity-achieving if its scattering matrix $\boldsymbol{\Theta}_F$ can be reconfigured such that $[\boldsymbol{\Theta}_F]_{N_S+\mathcal{N}_T,\mathcal{N}_S}=\bar{\mathbf{V}}$, for any given matrix $\bar{\mathbf{V}}\in\mathbb{C}^{N_T\times N_S}$ having orthogonal columns with unit $\ell_2$-norm.
This means that a transmitter-side \gls{milac} is capacity-achieving if the first $N_S$ columns of its scattering matrix can be reconfigured to $[\mathbf{0}_{N_S},\bar{\mathbf{V}}^T]^T$, for any possible $\bar{\mathbf{V}}$.
Formally, $\boldsymbol{\Theta}_F$ must satisfy
\begin{equation}
\boldsymbol{\Theta}_F
\begin{bmatrix}
\mathbf{I}_{N_S}\\
\mathbf{0}_{N_T\times N_S}
\end{bmatrix}
=
\begin{bmatrix}
\mathbf{0}_{N_S}\\
\bar{\mathbf{V}}
\end{bmatrix},\label{eq:tx-cond1}
\end{equation}
for any possible $\bar{\mathbf{V}}$, where the product $\boldsymbol{\Theta}_F[\mathbf{I}_{N_S},\mathbf{0}_{N_S\times N_T}]^T$ returns the first $N_S$ columns of $\boldsymbol{\Theta}_F$.
We have shown in \cite{ner25-3} that a fully-connected \gls{milac} (lossless and reciprocal) is capacity-achieving.
However, when considering different \gls{milac} architectures with a reduced number of admittance components, their susceptance matrix $\mathbf{B}_F$ must be such that $\mathbf{B}_F\in\mathcal{B}_{\mathcal{G}}$ and cannot be an arbitrary symmetric matrix.
As a consequence, the scattering matrix $\boldsymbol{\Theta}_F$, which is a function of $\mathbf{Y}_F=j\mathbf{B}_F$ as in \eqref{eq:TF(YF)}, cannot be an arbitrary unitary symmetric matrix, and condition \eqref{eq:tx-cond1} may not be achievable for any $\bar{\mathbf{V}}$.

Similarly, a \gls{milac} at the receiver is defined as capacity-achieving if its scattering matrix $\boldsymbol{\Theta}_G$ can be reconfigured such that $[\boldsymbol{\Theta}_G]_{N_R+\mathcal{N}_S,\mathcal{N}_R}=\bar{\mathbf{U}}^H$, for any $\bar{\mathbf{U}}\in\mathbb{C}^{N_R\times N_S}$ having orthogonal columns with unit $\ell_2$-norm.
In other words, a receiver-side \gls{milac} is capacity-achieving if the last $N_S$ rows of its scattering matrix can be reconfigured to $[\bar{\mathbf{U}}^H, \mathbf{0}_{N_S}]$, for any possible $\bar{\mathbf{U}}$, or, equivalently, the last $N_S$ columns of its scattering matrix can be reconfigured to $[\bar{\mathbf{U}}^H, \mathbf{0}_{N_S}]^T$, given the symmetry of $\boldsymbol{\Theta}_G$.
More precisely, $\boldsymbol{\Theta}_G$ must satisfy
\begin{equation}
\boldsymbol{\Theta}_G
\begin{bmatrix}
\mathbf{0}_{N_T\times N_S}\\
\mathbf{I}_{N_S}
\end{bmatrix}
=
\begin{bmatrix}
\bar{\mathbf{U}}^*\\
\mathbf{0}_{N_S}
\end{bmatrix},\label{eq:rx-cond1}
\end{equation}
for any possible $\bar{\mathbf{U}}$, where the product $\boldsymbol{\Theta}_G[\mathbf{0}_{N_S\times N_T},\mathbf{I}_{N_S}]^T$ returns the last $N_S$ columns of $\boldsymbol{\Theta}_G$.
As discussed for the transmitter-side \gls{milac}, a fully-connected \gls{milac} (lossless and reciprocal) is capacity-achieving \cite{ner25-3}.
However, other \gls{milac} architectures with a reduced number of admittance components, and hence reduced flexibility, could not be capacity-achieving.

\subsection{A Family of Capacity-Achieving MiLAC Architectures}

\begin{figure}[t]
\centering
\includegraphics[width=0.42\textwidth]{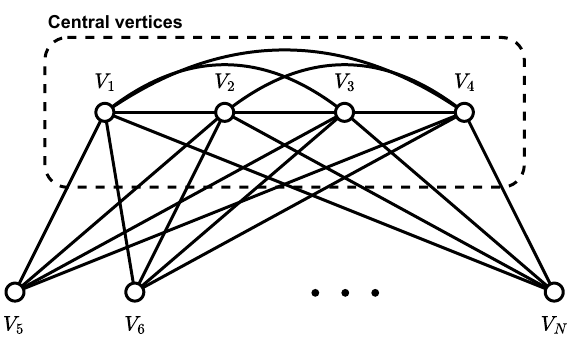}
\caption{Center graph with center size $Q=4$.}
\label{fig:center}
\end{figure}

As the circuit complexity of a fully-connected \gls{milac} becomes prohibitive when the number of antennas grows high, we are interested in alternative capacity-achieving \gls{milac} architectures with a reduced circuit complexity.
Before providing a sufficient condition for a \gls{milac} architecture to be capacity-achieving, we introduce the concept of center graph.
\begin{definition}
A center graph with center size $Q$ is a graph that contains $Q$ vertices, denoted as central vertices, each connected to all other vertices, while the remaining vertices are only connected to the central vertices.
\end{definition}
As an illustrative example, we report in Fig.~\ref{fig:center} a center graph on $N$ vertices with center size $Q=4$.
The central vertices are $V_1$, $V_2$, $V_3$, and $V_4$, which are connected to all vertices, while the other $N-4$ vertices are connected only to the central vertices.
We denote a \gls{milac} with associated graph being a center graph as a stem-connected \gls{milac}, following the terminology introduced in \gls{bd-ris} literature \cite{zho24,wu25}.

Based on this definition of center graph, the following proposition gives a sufficient condition for a transmitter-side \gls{milac} architecture to satisfy \eqref{eq:tx-cond1} for any $\bar{\mathbf{V}}$, i.e., for being capacity-achieving.
\begin{proposition}
A transmitter-side \gls{milac} with associated graph $\mathcal{G}$ is capacity-achieving if $\mathcal{G}$ is a center graph with center size $Q=2N_S-1$ and its central vertices include the $N_S$ vertices corresponding to the input ports of the \gls{milac}.
\label{pro:tx}
\end{proposition}
\begin{proof}
This proposition is proven in Section~\ref{sec:tx}, where we optimize \gls{milac} architectures satisfying this sufficient condition with a closed-form solution that fulfills \eqref{eq:tx-cond1} for any $\bar{\mathbf{V}}$.
\end{proof}
Fig.~\ref{fig:center-tx}(a) shows the graph corresponding to a stem-connected \gls{milac} at the transmitter, having center size $Q=2N_S-1$.
Note that the central vertices include all the $N_S$ vertices associated with the input ports (in red), i.e., $V_1,\ldots,V_{N_S}$.
In addition, the central vertices also include $N_S-1$ vertices associated with any $N_S-1$ output ports (in green), which are $V_{N_S+1},\ldots,V_{2N_S-1}$ in Fig.~\ref{fig:center-tx}(a).
Since in this graph the first $2N_S-1$ vertices are connected to all other vertices, the susceptance matrix of the corresponding \gls{milac} is constrained as $\mathbf{B}_F\in\mathcal{B}_F$, where
\begin{multline}
\mathcal{B}_F=\left\{\mathbf{B}\in\mathbb{R}^{(N_S+N_T)\times(N_S+N_T)}\mid\left[\mathbf{B}\right]_{n,m}=0,\;\forall n\neq m,\right.\\
n>2N_S-1,\;m>2N_S-1\Bigl\},\label{eq:B-set-tx}
\end{multline}
meaning that only the elements in the first $2N_S-1$ rows and columns and in the main diagonal are tunable.
A representation of such a matrix $\mathbf{B}_F$ is reported in Fig.~\ref{fig:center-tx}(b), where the tunable elements are highlighted in gray.
Note that, although in Fig.~\ref{fig:center-tx}(a) we have selected the output ports $V_{N_S+1},\ldots,V_{2N_S-1}$ to be connected to all the others, any set of $N_S-1$ output ports could be chosen.
Accordingly, the corresponding $N_S-1$ rows and columns of the matrix $\mathbf{B}_F$ in Fig.~\ref{fig:center-tx}(b) will be tunable.
This implies that there are multiple stem-connected \gls{milac} architectures that are capacity achieving, and that the number of such architectures is given by the binomial coefficient $\binom{N_T}{N_S-1}$.
If $N_S=1$, only one architecture is optimal (having a single central vertex corresponding to the only input port of the \gls{milac}) since $\binom{N_T}{0}=1$.

\begin{figure}[t]
\centering
\subfigure[Center graph with center size $Q=2N_S-1$.]{
\includegraphics[width=0.44\textwidth]{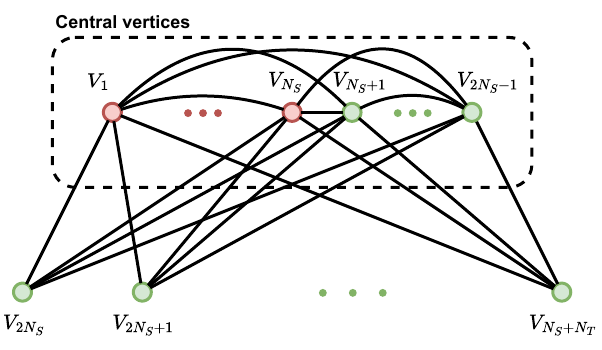}}
\subfigure[Susceptance matrix $\mathbf{B}_F$.]{
\includegraphics[width=0.22\textwidth]{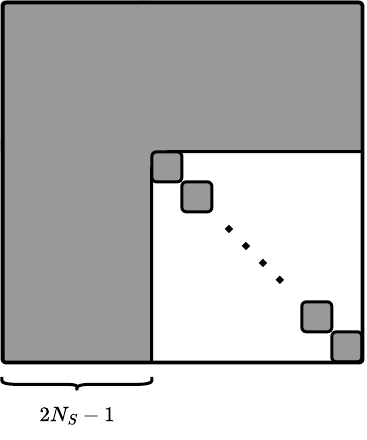}}
\caption{(a) Graph and (b) susceptance matrix (with tunable elements in gray) of a capacity-achieving stem-connected MiLAC at the transmitter.}
\label{fig:center-tx}
\end{figure}

\begin{figure}[t]
\centering
\subfigure[Center graph with center size $Q=2N_S-1$.]{
\includegraphics[width=0.44\textwidth]{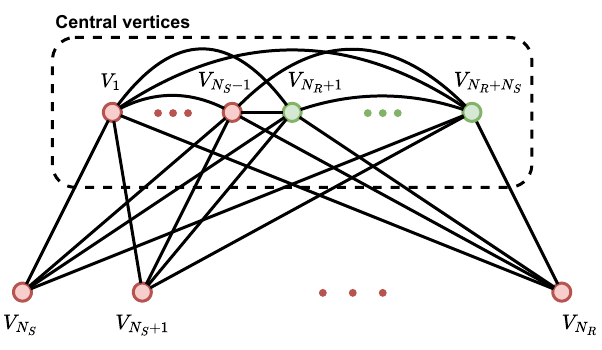}}
\subfigure[Susceptance matrix $\mathbf{B}_G$.]{
\includegraphics[width=0.22\textwidth]{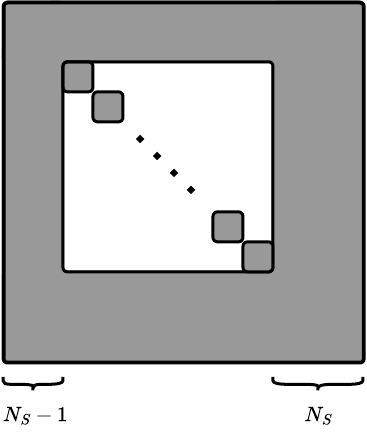}}
\caption{(a) Graph and (b) susceptance matrix (with tunable elements in gray) of a capacity-achieving stem-connected MiLAC at the receiver.}
\label{fig:center-rx}
\end{figure}

\begin{table*}[t]
\centering
\caption{Fully- and stem-connected MiLAC architectures.}
\label{tab:milac}
\begin{tabular}{@{}llll@{}}
\toprule
MiLAC architecture & Associated graph & Circuit complexity$^\star$ & Description$^\star$\\
\midrule
Fully-connected MiLAC & Complete graph & $N_C^\text{Fully}=\left(N_S+N_T\right)\left(N_S+N_T+1\right)/2$ &
\begin{tabular}{@{}l@{}}All $N_S+N_T$ MiLAC ports are\\connected to each other\end{tabular}\\
\midrule
Stem-connected MiLAC & Center graph & $N_C^\text{Stem}=N_S\left(2N_T+1\right)$ &
\begin{tabular}{@{}l@{}}Only $2N_S-1$ MiLAC ports are\\connected to all the others\end{tabular}\\
\bottomrule
\end{tabular}
\newline\newline
\footnotesize $^\star$ For the stem-connected MiLAC, we consider $Q=2N_S-1$. The same discussion applies to a receiver-side MiLAC by replacing $N_T$ with $N_R$.
\end{table*}

As discussed for a transmitter-side \gls{milac}, the following proposition gives a sufficient condition for a receiver-side \gls{milac} architecture to satisfy \eqref{eq:rx-cond1} for any $\bar{\mathbf{U}}$, i.e., for being capacity-achieving.
\begin{proposition}
A receiver-side \gls{milac} with associated graph $\mathcal{G}$ is capacity-achieving if $\mathcal{G}$ is a center graph with center size $Q=2N_S-1$ and its central vertices include the $N_S$ vertices corresponding to the output ports of the \gls{milac}.
\label{pro:rx}
\end{proposition}
\begin{proof}
Similar to Proof of Proposition~\ref{pro:tx}, we prove this proposition in Section~\ref{sec:rx} by optimizing these \gls{milac} architectures with a closed-form solution that achieves capacity.
\end{proof}
Fig.~\ref{fig:center-rx}(a) shows the graph corresponding to a stem-connected \gls{milac} at the receiver, with center size $Q=2N_S-1$.
The central vertices include all the $N_S$ vertices associated with the output ports (in green), i.e., $V_{N_R+1},\ldots,V_{N_R+N_S}$ .
In addition, the central vertices also include $N_S-1$ vertices associated with any $N_S-1$ input ports (in red), namely $V_{1},\ldots,V_{N_S-1}$ in Fig.~\ref{fig:center-rx}(a).
Since in this graph the first $N_S-1$ and last $N_S$ vertices are connected to all other vertices, the susceptance matrix of the corresponding \gls{milac} is constrained as $\mathbf{B}_G\in\mathcal{B}_G$, where
\begin{multline}
\mathcal{B}_G=\left\{\mathbf{B}\in\mathbb{R}^{(N_R+N_S)\times (N_R+N_S)}\mid\left[\mathbf{B}\right]_{n,m}=0,\;\forall n\neq m,\right.\\
N_S-1<n\leq N_R,\;N_S-1<m\leq N_R\Bigl\},\label{eq:B-set-rx}
\end{multline}
meaning that only the elements in the first $N_S-1$ and last $N_S$ rows and columns and in the main diagonal are tunable.
A representation of such a matrix $\mathbf{B}_G$ is reported in Fig.~\ref{fig:center-rx}(b), where the tunable elements are highlighted in gray.
Note that, although in Fig.~\ref{fig:center-rx}(a) we have selected the input ports $V_{1},\ldots,V_{N_S-1}$ to be connected to all the others, any set of $N_S-1$ input ports could be chosen\footnote{We have selected the input ports $V_{1},\ldots,V_{N_S-1}$ because the optimization of such a receiver-side \gls{milac} is analogous to that of the transmitter-side \gls{milac} shown in Fig.~\ref{fig:center-tx}, as it will be clarified in Section~\ref{sec:rx}.}.
Accordingly, the corresponding $N_S-1$ rows and columns of the matrix $\mathbf{B}_G$ in Fig.~\ref{fig:center-rx}(b) will be tunable.
This means that there are multiple stem-connected \gls{milac} architectures that are capacity achieving, specifically $\binom{N_R}{N_S-1}$.
If $N_S=1$, only one architecture is optimal as $\binom{N_R}{0}=1$.

A stem-connected \gls{milac} can achieve capacity with a significantly reduced circuit complexity compared to a fully-connected \gls{milac}.
To quantify the circuit complexity of a stem-connected \gls{milac}, we recall that its complexity is given by $N_C=N_V+N_E$, where $N_V$ and $N_E$ are the number of vertices and edges of its associated graph, which is a center graph with center size $Q$.
Since the number of edges in a center graph on $N_V$ vertices and having center size $Q$ is $N_E=Q(Q-1)/2+Q(N_V-Q)$, the complexity of a stem-connected \gls{milac} is given by
\begin{align}
N_C^\text{Stem}
&=N_V+\frac{Q\left(Q-1\right)}{2}+Q\left(N_V-Q\right)\\
&=\frac{\left(Q+1\right)\left(2N_V-Q\right)}{2}.\label{eq:center-complexity}
\end{align}
Thus, in the case of a capacity-achieving stem-connected \gls{milac} at the transmitter, we can substitute $N_V=N_S+N_T$ and $Q=2N_S-1$ into \eqref{eq:center-complexity} and obtain
\begin{equation}
N_C^\text{Stem}=N_S\left(2N_T+1\right),\label{eq:stem-tx-complexity}
\end{equation}
growing with $\mathcal{O}(N_SN_T)$, i.e., linearly with the number of antennas $N_T$.
A similar discussion holds for a receiver-side \gls{milac} and, by replacing $N_T$ in \eqref{eq:stem-tx-complexity} with $N_R$, we obtain that its circuit complexity scales with $\mathcal{O}(N_SN_R)$, i.e., linearly with $N_R$.
We summarize the main properties of fully- and stem-connected \glspl{milac} in Tab.~\ref{tab:milac}.

\begin{remark}
The capacity-achieving stem-connected \gls{milac} architectures introduced in Propositions~\ref{pro:tx} and \ref{pro:rx} have a center size $Q=2N_S-1$.
Note that such a value of the center size is aligned with the center size of optimal stem-connected \gls{ris} architectures proposed in \cite{zho24} and analyzed in \cite{wu25}.
In detail, a sufficient condition for stem-connected \glspl{ris} to achieve maximum performance is to have a center size $Q=2K-1$, with $K$ denoting the degrees of freedom, or multiplexing gain, of the \gls{ris}-aided \gls{mimo} channel, as observed in \cite{zho24} and proven in \cite[Corollary~2]{wu25}.
Similarly, the capacity-achieving stem-connected \glspl{milac} in Propositions~\ref{pro:tx} and \ref{pro:rx} have a center size $Q=2N_S-1$, as the number of streams $N_S$ is the multiplexing gain of the \gls{milac}-aided \gls{mimo} channel.
The value $Q=2N_S-1$ is sufficient to achieve capacity, but not proved to be necessary, as proving necessity would require showing that no other center graph with fewer than $2N_S-1$ central vertices can achieve capacity.
\end{remark}

\begin{remark}
The capacity-achieving stem-connected \gls{milac} architectures in Proposition~\ref{pro:tx} have associated graphs including the $N_S$ vertices corresponding to the input \gls{milac} ports within their central vertices.
Likewise, the \gls{milac} architectures in Proposition~\ref{pro:rx} have associated graphs including the $N_S$ vertices corresponding to the output \gls{milac} ports within their central vertices.
Note that this constraint is imposed since it is required by our capacity-achieving closed-form solutions, as it will be clarified in Sections~\ref{sec:tx} and \ref{sec:rx}.
Thus, this constraint is not proven to be necessary to achieve capacity.
\end{remark}

\section{Optimization of Capacity-Achieving Transmitter-Side MiLAC Architectures}
\label{sec:tx}

We have introduced stem-connected \gls{milac} architectures and claimed in Proposition~\ref{pro:tx} that they are capacity-achieving.
In this section, we prove Proposition~\ref{pro:tx} by optimizing a stem-connected \gls{milac} at the transmitter through a closed-form that achieves capacity.
To simplify the notation, only in this section, we denote the scattering, admittance, and susceptance matrices of the considered \gls{milac} at the transmitter as $\boldsymbol{\Theta}$, $\mathbf{Y}$, and $\mathbf{B}$ instead of $\boldsymbol{\Theta}_F$, $\mathbf{Y}_F$, and $\mathbf{B}_F$.

\subsection{Reformulation of the Capacity-Achieving Condition}

By recalling that $\boldsymbol{\Theta}$ can be expressed as a function of the admittance matrix $\mathbf{Y}=j\mathbf{B}$ as in \eqref{eq:TF(YF)}, the capacity-achieving condition in \eqref{eq:tx-cond1} can be reformulated as
\begin{equation}
\left(Y_0\mathbf{I}-j\mathbf{B}\right)
\begin{bmatrix}
\mathbf{I}_{N_S}\\
\mathbf{0}_{N_T\times N_S}
\end{bmatrix}
=
\left(Y_0\mathbf{I}+j\mathbf{B}\right)
\begin{bmatrix}
\mathbf{0}_{N_S}\\
\bar{\mathbf{V}}
\end{bmatrix},\label{eq:tx-cond2}
\end{equation}
as a function of $\mathbf{B}$ rather than $\boldsymbol{\Theta}$.
After algebraic computations, from \eqref{eq:tx-cond2} we obtain
\begin{equation}
j\mathbf{B}
\begin{bmatrix}
\mathbf{I}_{N_S}\\
\bar{\mathbf{V}}
\end{bmatrix}
=Y_0
\begin{bmatrix}
\mathbf{I}_{N_S}\\
-\bar{\mathbf{V}}
\end{bmatrix},
\end{equation}
or, equivalently,
\begin{equation}
j
\begin{bmatrix}
\mathbf{I}_{N_S} & \bar{\mathbf{V}}^T
\end{bmatrix}
\mathbf{B}
=Y_0
\begin{bmatrix}
\mathbf{I}_{N_S} & -\bar{\mathbf{V}}^T
\end{bmatrix},\label{eq:tx-cond3}
\end{equation}
since $\mathbf{B}$ is symmetric.
The condition in \eqref{eq:tx-cond3} is a linear matrix equation with real unknowns, i.e., the entries of the real matrix $\mathbf{B}$, and complex coefficients, i.e., the entries of the complex matrices $j[\mathbf{I}_{N_S},\bar{\mathbf{V}}^T]$ and $Y_0[\mathbf{I}_{N_S},-\bar{\mathbf{V}}^T]$.
To transform \eqref{eq:tx-cond3} into a linear matrix equation with real coefficients, we equivalently rewrite it as
\begin{equation}
\begin{bmatrix}
\Re\left\{j
\begin{bmatrix}
\mathbf{I}_{N_S} & \bar{\mathbf{V}}^T
\end{bmatrix}\right\}\\
\Im\left\{j
\begin{bmatrix}
\mathbf{I}_{N_S} & \bar{\mathbf{V}}^T
\end{bmatrix}\right\}
\end{bmatrix}
\mathbf{B}=
\begin{bmatrix}
\Re\left\{Y_0
\begin{bmatrix}
\mathbf{I}_{N_S} & -\bar{\mathbf{V}}^T
\end{bmatrix}\right\}\\
\Im\left\{Y_0
\begin{bmatrix}
\mathbf{I}_{N_S} & -\bar{\mathbf{V}}^T
\end{bmatrix}\right\}
\end{bmatrix},
\end{equation}
which can be expressed more concisely as
\begin{equation}
\begin{bmatrix}
\mathbf{0}_{N_S} & -\Im\left\{\bar{\mathbf{V}}\right\}^T\\
\mathbf{I}_{N_S} & \Re\left\{\bar{\mathbf{V}}\right\}^T
\end{bmatrix}
\mathbf{B}=Y_0
\begin{bmatrix}
\mathbf{I}_{N_S} & -\Re\left\{\bar{\mathbf{V}}\right\}^T\\
\mathbf{0}_{N_S} & -\Im\left\{\bar{\mathbf{V}}\right\}^T
\end{bmatrix}.\label{eq:tx-cond4}
\end{equation}
To find a matrix $\mathbf{B}$ satisfying \eqref{eq:tx-cond4}, we partition $\mathbf{B}$ as 
\begin{equation}
\mathbf{B}=
\begin{bmatrix}
\mathbf{B}_{11} & \mathbf{B}_{12}\\
\mathbf{B}_{21} & \mathbf{B}_{22}
\end{bmatrix},\label{eq:B-part}
\end{equation}
where $\mathbf{B}_{11}\in\mathbb{R}^{N_S\times N_S}$, $\mathbf{B}_{12}\in\mathbb{R}^{N_S\times N_T}$, $\mathbf{B}_{21}\in\mathbb{R}^{N_T\times N_S}$, and $\mathbf{B}_{22}\in\mathbb{R}^{N_T\times N_T}$.
Thus, by introducing $\mathbf{R}\in\mathbb{R}^{N_S\times N_T}$ and $\mathbf{J}\in\mathbb{R}^{N_S\times N_T}$ as $\mathbf{R}=\Re\left\{\bar{\mathbf{V}}\right\}^T$ and $\mathbf{J}=\Im\left\{\bar{\mathbf{V}}\right\}^T$ to alleviate the notation, \eqref{eq:tx-cond4} becomes
\begin{equation}
\begin{bmatrix}
\mathbf{0}_{N_S} & -\mathbf{J}\\
\mathbf{I}_{N_S} & \mathbf{R}
\end{bmatrix}
\begin{bmatrix}
\mathbf{B}_{11} & \mathbf{B}_{12}\\
\mathbf{B}_{21} & \mathbf{B}_{22}
\end{bmatrix}
=Y_0
\begin{bmatrix}
\mathbf{I}_{N_S} & -\mathbf{R}\\
\mathbf{0}_{N_S} & -\mathbf{J}
\end{bmatrix}.\label{eq:tx-cond5}
\end{equation}
To formalize the constraints on the blocks $\mathbf{B}_{11}$, $\mathbf{B}_{12}$, $\mathbf{B}_{21}$, and $\mathbf{B}_{22}$, we assume with no loss of generality that the considered stem-connected \gls{milac} has the $N_S-1$ output ports $V_{N_S+1},\ldots,V_{2N_S-1}$ connected to all the others, as Fig.~\ref{fig:center-tx}(a).
Accordingly, we have $\mathbf{B}\in\mathcal{B}_F$, where $\mathcal{B}_F$ is given by \eqref{eq:B-set-tx}.
This means that the only constraints on $\mathbf{B}_{11}$, $\mathbf{B}_{12}$, and $\mathbf{B}_{21}$ are $\mathbf{B}_{11}=\mathbf{B}_{11}^T$ and $\mathbf{B}_{21}=\mathbf{B}_{12}^T$ (due to the symmetry of $\mathbf{B}$, and since the $N_S$ input ports of the \gls{milac} are connected to all the others as required by Proposition~\ref{pro:tx}), while $\mathbf{B}_{22}$ satisfies $\mathbf{B}_{22}=\mathbf{B}_{22}^T$ and $\mathbf{B}_{22}\in\mathcal{B}_{22}$, where
\begin{multline}
\mathcal{B}_{22}=\left\{\mathbf{B}\in\mathbb{R}^{N_T\times N_T}\mid\left[\mathbf{B}\right]_{n,m}=0,\;\forall n\neq m,\right.\\
n>N_S-1,\;m>N_S-1\Bigl\}.\label{eq:B-set-22}
\end{multline}

\begin{figure}[t]
\centering
\includegraphics[height=0.44\textwidth]{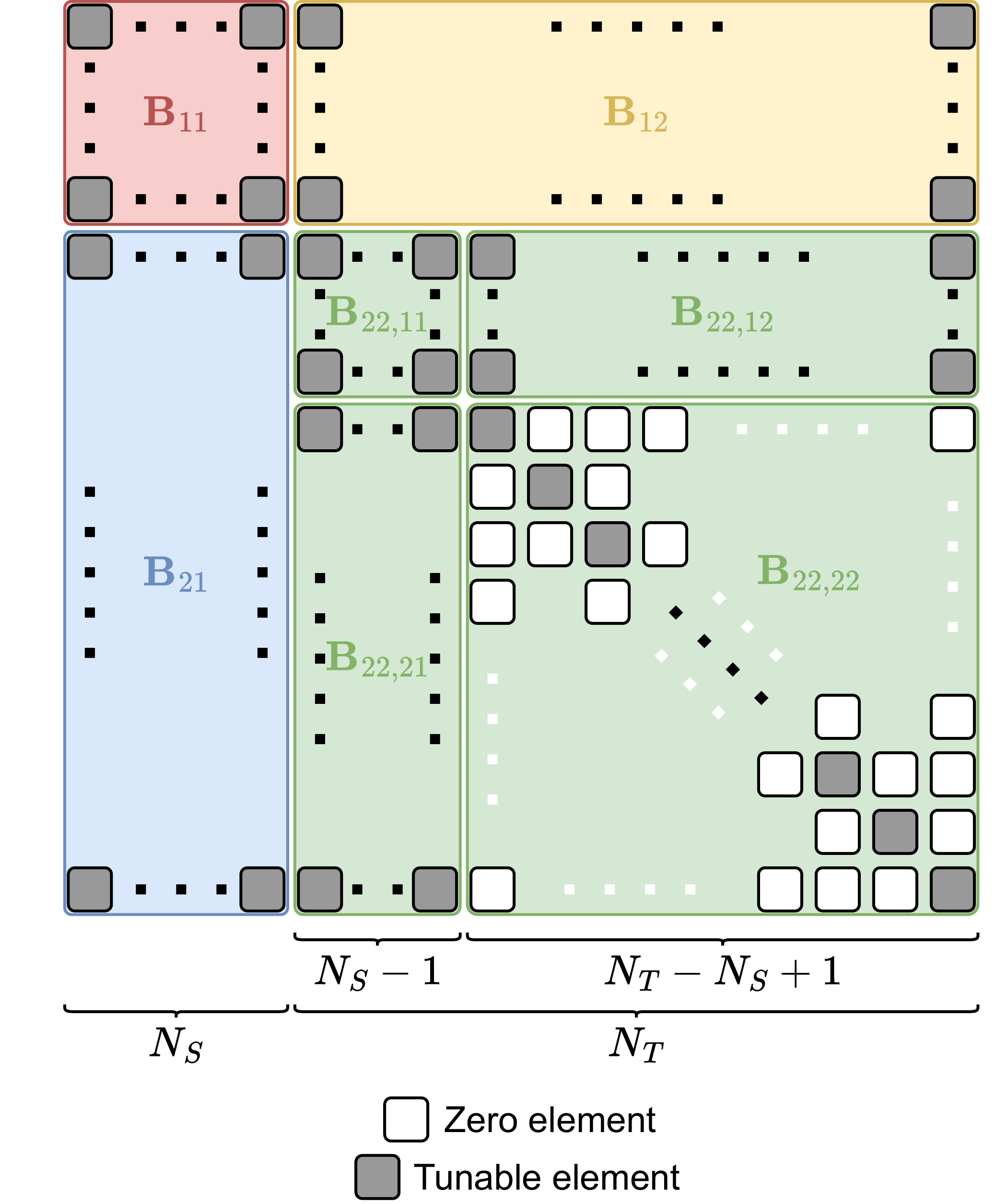}
\caption{Admittance matrix $\mathbf{B}$ of a stem-connected MiLAC at the
transmitter partitioned into seven matrices, where $\mathbf{B}_{22,22}$ is diagonal.}
\label{fig:B-tx}
\end{figure}

As a result, the problem of optimizing a stem-connected \gls{milac} through a capacity-achieving solution can be formalized as a feasibility check problem, where we need to find a valid $\mathbf{B}$ partitioned as in \eqref{eq:B-part} whose blocks satisfy the capacity-achieving condition \eqref{eq:tx-cond5}.
Formally, we need to solve
\begin{align}
\mathrm{find}\;\;
&\mathbf{B}_{11},\;\mathbf{B}_{12},\;\mathbf{B}_{21},\;\mathbf{B}_{22}\label{eq:tx-prob}\\
\mathsf{\mathrm{s.t.}}\;\;\;
&\mathbf{B}_{11}=\mathbf{B}_{11}^T,\;\mathbf{B}_{21}=\mathbf{B}_{12}^T,\;\mathbf{B}_{22}=\mathbf{B}_{22}^T,\label{eq:tx-symm}\\
&\mathbf{B}_{22}\in\mathcal{B}_{22},\;\eqref{eq:B-set-22}\label{eq:tx-graph}\\
&\mathbf{J}\mathbf{B}_{21}=-Y_0\mathbf{I}_{N_S},\label{eq:tx-case1}\\
&\mathbf{J}\mathbf{B}_{22}=Y_0\mathbf{R},\label{eq:tx-case2}\\
&\mathbf{B}_{11}+\mathbf{R}\mathbf{B}_{21}=\mathbf{0}_{N_S},\label{eq:tx-case3}\\
&\mathbf{B}_{12}+\mathbf{R}\mathbf{B}_{22}=-Y_0\mathbf{J},\label{eq:tx-case4}
\end{align}
where \eqref{eq:tx-symm} and \eqref{eq:tx-graph} follow from the constraints on $\mathbf{B}$, and \eqref{eq:tx-case1}-\eqref{eq:tx-case4} follow from the capacity-achieving condition \eqref{eq:tx-cond5}.
In the following subsections, we solve \eqref{eq:tx-prob}-\eqref{eq:tx-case4} by first finding a valid block $\mathbf{B}_{22}$ fulfilling \eqref{eq:tx-case2}, and then finding valid blocks $\mathbf{B}_{11}$, $\mathbf{B}_{12}$, and $\mathbf{B}_{21}$ fulfilling the other constraints in \eqref{eq:tx-prob}-\eqref{eq:tx-case4}.

\subsection{Block $\mathbf{B}_{22}$}
\label{sec:B22}

We first find the block $\mathbf{B}_{22}$ subject to $\mathbf{B}_{22}=\mathbf{B}_{22}^T$ and $\mathbf{B}_{22}\in\mathcal{B}_{22}$ by solving \eqref{eq:tx-case2}.
To this end, we conveniently partition the matrix $\mathbf{B}_{22}$ as
\begin{equation}
\mathbf{B}_{22}=
\begin{bmatrix}
\mathbf{B}_{22,11} & \mathbf{B}_{22,12}\\
\mathbf{B}_{22,21} & \mathbf{B}_{22,22}
\end{bmatrix},\label{eq:B22-part}
\end{equation}
where $\mathbf{B}_{22,11}\in\mathbb{R}^{(N_S-1)\times (N_S-1)}$ is symmetric, $\mathbf{B}_{22,12}\in\mathbb{R}^{(N_S-1)\times (N_T-N_S+1)}$ and $\mathbf{B}_{22,21}\in\mathbb{R}^{(N_T-N_S+1)\times (N_S-1)}$ are such that $\mathbf{B}_{22,21}=\mathbf{B}_{22,12}^T$, and $\mathbf{B}_{22,22}\in\mathbb{R}^{(N_T-N_S+1)\times (N_T-N_S+1)}$ is diagonal since $\mathbf{B}_{22}\in\mathcal{B}_{22}$.
Note that by substituting \eqref{eq:B22-part} into \eqref{eq:B-part}, we obtain that the matrix $\mathbf{B}$ is partitioned into a total of seven matrices, as graphically represented in Fig.~\ref{fig:B-tx}.

By considering the partition in \eqref{eq:B22-part} and the partitions $\mathbf{R}=[\mathbf{R}_1,\mathbf{R}_2]$, with $\mathbf{R}_1\in\mathbb{R}^{N_S\times (N_S-1)}$ and $\mathbf{R}_2\in\mathbb{R}^{N_S\times (N_T-N_S+1)}$ and $\mathbf{J}=[\mathbf{J}_1,\mathbf{J}_2]$, with $\mathbf{J}_1\in\mathbb{R}^{N_S\times (N_S-1)}$ and $\mathbf{J}_2\in\mathbb{R}^{N_S\times (N_T-N_S+1)}$, \eqref{eq:tx-case2} can be rewritten as
\begin{equation}
\begin{bmatrix}
\mathbf{J}_1 & \mathbf{J}_2
\end{bmatrix}
\begin{bmatrix}
\mathbf{B}_{22,11} & \mathbf{B}_{22,12}\\
\mathbf{B}_{22,21} & \mathbf{B}_{22,22}
\end{bmatrix}
=Y_0
\begin{bmatrix}
\mathbf{R}_1 & \mathbf{R}_2
\end{bmatrix},
\end{equation}
which is equivalent to a system of two linear matrix equations 
\begin{numcases}{}
\mathbf{J}_1\mathbf{B}_{22,11} + \mathbf{J}_2\mathbf{B}_{22,21} = Y_0\mathbf{R}_1,\label{eq:B22-cond-1}\\
\mathbf{J}_1\mathbf{B}_{22,12} + \mathbf{J}_2\mathbf{B}_{22,22} = Y_0\mathbf{R}_2.\label{eq:B22-cond-2}
\end{numcases}
In the following, we find the values for the blocks of $\mathbf{B}_{22}$ fulfilling this system of equations.
We first set the blocks $\mathbf{B}_{22,12}$, $\mathbf{B}_{22,21}$, and $\mathbf{B}_{22,22}$, and then optimize $\mathbf{B}_{22,11}$ accordingly.

\subsubsection{Blocks $\mathbf{B}_{22,12}$, $\mathbf{B}_{22,21}$, and $\mathbf{B}_{22,22}$}
\label{sec:B22-blocks}

Denote as $\mathbf{r}_i\in\mathbb{R}^{N_S\times1}$ the $i$th column of $\mathbf{R}$ and as $\mathbf{j}_i\in\mathbb{R}^{N_S\times1}$ the $i$th column of $\mathbf{J}$, for $i=1,\ldots,N_T$.
Thus, we claim that the $i$th diagonal element of the diagonal matrix $\mathbf{B}_{22,22}$ is given by
\begin{equation}
\left[\mathbf{B}_{22,22}\right]_{i,i}=Y_0
\left[\left[\mathbf{J}_1,\mathbf{j}_{N_S-1+i}\right]^{-1}\mathbf{r}_{N_S-1+i}\right]_{N_S},\label{eq:B2222ii}
\end{equation}
for $i=1,\ldots,N_T-N_S+1$.
In addition, we claim that the $(j,i)$th element of $\mathbf{B}_{22,12}$ is given by
\begin{equation}
\left[\mathbf{B}_{22,12}\right]_{j,i}=Y_0
\left[\left[\mathbf{J}_1,\mathbf{j}_{N_S-1+i}\right]^{-1}\mathbf{r}_{N_S-1+i}\right]_{j},\label{eq:B2212ji}
\end{equation}
for $j=1,\ldots,N_S-1$ and $i=1,\ldots,N_T-N_S+1$.
These claims can be immediately verified by checking that \eqref{eq:B2222ii} and \eqref{eq:B2212ji} satisfy \eqref{eq:B22-cond-2}, i.e.,
\begin{equation}
\left[\mathbf{J}_1\mathbf{B}_{22,12} + \mathbf{J}_2\mathbf{B}_{22,22}\right]_{j,i}
=Y_0
\left[\mathbf{R}_2\right]_{j,i},\label{eq:B22-cond-2-2}
\end{equation}
for $j=1,\ldots,N_S$ and $i=1,\ldots,N_T-N_S+1$.
By substituting \eqref{eq:B2222ii} and \eqref{eq:B2212ji} into the left side of \eqref{eq:B22-cond-2-2}, we have
\begin{multline}
\left[\mathbf{J}_1\mathbf{B}_{22,12} + \mathbf{J}_2\mathbf{B}_{22,22}\right]_{j,i}
=\left[\mathbf{J}_1\mathbf{B}_{22,12}\right]_{j,i} + \left[\mathbf{J}_2\mathbf{B}_{22,22}\right]_{j,i}\\
=Y_0\left[\mathbf{J}_1\right]_{j,:}
\left[\left[\mathbf{J}_1,\mathbf{j}_{N_S-1+i}\right]^{-1}\mathbf{r}_{N_S-1+i}\right]_{1:N_S-1}\\
+Y_0\left[\mathbf{J}_2\right]_{j,i}
\left[\left[\mathbf{J}_1,\mathbf{j}_{N_S-1+i}\right]^{-1}\mathbf{r}_{N_S-1+i}\right]_{N_S}\\
=Y_0\left[\mathbf{J}_1\right]_{j,:}
\left[\left[\mathbf{J}_1,\mathbf{j}_{N_S-1+i}\right]^{-1}\right]_{1:N_S-1,:}\mathbf{r}_{N_S-1+i}\\
+Y_0\left[\mathbf{j}_{N_S-1+i}\right]_{j}
\left[\left[\mathbf{J}_1,\mathbf{j}_{N_S-1+i}\right]^{-1}\right]_{N_S,:}\mathbf{r}_{N_S-1+i}\\
=Y_0\left[\mathbf{J}_1,\mathbf{j}_{N_S-1+i}\right]_{j,:}
\left[\mathbf{J}_1,\mathbf{j}_{N_S-1+i}\right]^{-1}\mathbf{r}_{N_S-1+i}\\
=Y_0\left[\mathbf{r}_{N_S-1+i}\right]_{j}
=Y_0\left[\mathbf{R}_2\right]_{j,i},
\end{multline}
confirming that \eqref{eq:B2222ii} and \eqref{eq:B2212ji} are optimal solutions as they satisfy \eqref{eq:B22-cond-2}.
Since $\mathbf{B}_{22,12}$ is given by \eqref{eq:B2212ji}, we also know $\mathbf{B}_{22,21}$, which can be computed as $\mathbf{B}_{22,21}=\mathbf{B}_{22,12}^T$.

\subsubsection{Block $\mathbf{B}_{22,11}$}
\label{sec:B22-11}

\begin{algorithm}[t]
\begin{algorithmic}[1]
\setstretch{1.2}
\REQUIRE $\bar{\mathbf{V}}$.
\ENSURE $\mathbf{B}$.
\STATEx{$\mathbf{R}=\left[\mathbf{R}_1,\mathbf{R}_2\right]=\Re\left\{\bar{\mathbf{V}}\right\}^T$,
$\mathbf{J}=\left[\mathbf{J}_1,\mathbf{J}_2\right]=\Im\left\{\bar{\mathbf{V}}\right\}^T$,}
\STATEx{$\mathbf{r}_i=[\mathbf{R}]_{:,i}$,
$\mathbf{j}_i=[\mathbf{J}]_{:,i}$,}
\STATEx{$\mathbf{J}_1=\left[\mathbf{U}_{J,1},\mathbf{u}_{J,2}\right]\left[\mathbf{\Sigma}_{J},\mathbf{0}_{(N_S-1)\times1}\right]^T\mathbf{V}_{J}^T$.}
\STATE{\textcolor{mygreen}{$\left[\mathbf{B}_{22,22}\right]_{i,i}=Y_0
\left[\left[\mathbf{J}_1,\mathbf{j}_{N_S-1+i}\right]^{-1}\mathbf{r}_{N_S-1+i}\right]_{N_S}$.}}
\STATE{\textcolor{mygreen}{$\left[\mathbf{B}_{22,12}\right]_{j,i}=Y_0
\left[\left[\mathbf{J}_1,\mathbf{j}_{N_S-1+i}\right]^{-1}\mathbf{r}_{N_S-1+i}\right]_{j}$.}}
\STATE{\textcolor{mygreen}{$\mathbf{B}_{22,21}=\mathbf{B}_{22,12}^T$.}}
\STATE{\textcolor{mygreen}{$\mathbf{B}_{22,11} = \mathbf{V}_{J}\mathbf{\Sigma}_{J}^{-1}\mathbf{U}_{J,1}^T\left(Y_0\mathbf{R}_1-\mathbf{J}_2\mathbf{B}_{22,21}\right)$.}}
\STATEx{$\mathbf{B}_{22}=
\left[\left[\mathbf{B}_{22,11},\mathbf{B}_{22,12}\right]^T,
\left[\mathbf{B}_{22,21},\mathbf{B}_{22,22}\right]^T\right]^T$.}
\STATE{\textcolor{myyellow}{$\mathbf{B}_{12}=-Y_0\mathbf{J}-\mathbf{R}\mathbf{B}_{22}$.}}
\STATE{\textcolor{myblue}{$\mathbf{B}_{21}=\mathbf{B}_{12}^T$.}}
\STATE{\textcolor{myred}{$\mathbf{B}_{11}=-\mathbf{R}\mathbf{B}_{21}$.}}
\STATEx{$\mathbf{B}=
\left[\left[\mathbf{B}_{11},\mathbf{B}_{12}\right]^T,
\left[\mathbf{B}_{21},\mathbf{B}_{22}\right]^T\right]^T$.}
\end{algorithmic}
\caption{Optimization of a stem-connected MiLAC at the transmitter.}
\label{alg:tx}
\end{algorithm}

We now need to find a symmetric matrix $\mathbf{B}_{22,11}$ such that \eqref{eq:B22-cond-1} is satisfied, i.e.,
\begin{equation}
\mathbf{J}_1\mathbf{B}_{22,11}=Y_0\mathbf{R}_1-\mathbf{J}_2\mathbf{B}_{22,21},\label{eq:B22-cond-1-2}
\end{equation}
where $\mathbf{B}_{22,21}$ is given as $\mathbf{B}_{22,21}=\mathbf{B}_{22,12}^T$, with $\mathbf{B}_{22,12}$ set according to \eqref{eq:B2212ji}.
To solve the linear matrix equation in \eqref{eq:B22-cond-1-2} where the unknown matrix is symmetric, we introduce the following result.
\begin{proposition}
Given two matrices $\mathbf{A}\in\mathbb{R}^{M\times N}$ and $\mathbf{C}\in\mathbb{R}^{M\times N}$, consider the linear matrix equation
\begin{equation}
\mathbf{A}\mathbf{X}=\mathbf{C},\label{eq:symm-eq1}
\end{equation}
where the unknown matrix $\mathbf{X}\in\mathbb{R}^{N\times N}$ is constrained to be symmetric.
The \gls{svd} of $\mathbf{A}$ is
\begin{equation}
\mathbf{A}=
\mathbf{U}
\begin{bmatrix}
\mathbf{\Sigma} & \mathbf{0}_{R\times (N-R)}\\
\mathbf{0}_{(M-R)\times R} & \mathbf{0}_{(M-R)\times (N-R)}
\end{bmatrix}
\mathbf{V}^T,
\end{equation}
where $\mathbf{U}\in\mathbb{R}^{M\times M}$ contains the left singular vectors of $\mathbf{A}$, $\mathbf{\Sigma}\in\mathbb{R}^{R\times R}$ is a diagonal matrix containing the (positive) singular values of $\mathbf{A}$, with $R$ being the rank of $\mathbf{A}$, and $\mathbf{V}\in\mathbb{R}^{N\times N}$ contains the right singular vectors of $\mathbf{A}$.
We partition $\mathbf{U}=[\mathbf{U}_1,\mathbf{U}_2]$ with $\mathbf{U}_1\in\mathbb{R}^{M\times R}$ and $\mathbf{U}_2\in\mathbb{R}^{M\times (M-R)}$, and $\mathbf{V}=[\mathbf{V}_1,\mathbf{V}_2]$ with $\mathbf{V}_1\in\mathbb{R}^{N\times R}$ and $\mathbf{V}_2\in\mathbb{R}^{N\times (N-R)}$.
Thus, equation \eqref{eq:symm-eq1} has a symmetric solution if and only if $\mathbf{A}\mathbf{C}^T=\mathbf{C}\mathbf{A}^T$ and $\mathbf{U}_2^T\mathbf{C}=\mathbf{0}_{(M-R)\times N}$.
In this case, all solutions can be expressed as
\begin{equation}
\mathbf{X}=
\mathbf{V}_1\mathbf{\Sigma}^{-1}\mathbf{U}_1^T\mathbf{C}
+\mathbf{V}_2\mathbf{V}_2^T\mathbf{C}^T\mathbf{U}_1\mathbf{\Sigma}^{-1}\mathbf{V}_1^T
+\mathbf{V}_2\mathbf{G}\mathbf{V}_2^T,\label{eq:symm-sol1}
\end{equation}
where $\mathbf{G}\in\mathbb{R}^{(N-R)\times(N-R)}$ is any symmetric matrix.
\label{pro:symm-eq}
\end{proposition}
\begin{proof}
This proposition follows from \cite[Theorem~1]{hua90}.
\end{proof}
From Proposition~\ref{pro:symm-eq}, we can prove a useful corollary.
\begin{corollary}
Consider the setup of Proposition~\ref{pro:symm-eq} with $M=N+1$ and $\mathbf{A}$ having full column rank, i.e., having rank $R=N$.
Thus, equation \eqref{eq:symm-eq1} has a symmetric solution if and only if $\mathbf{A}\mathbf{C}^T=\mathbf{C}\mathbf{A}^T$.
In this case, the only solution is given by
\begin{equation}
\mathbf{X}=\mathbf{V}\mathbf{\Sigma}^{-1}\mathbf{U}_1^T\mathbf{C}.\label{eq:symm-sol2}
\end{equation}
\label{cor:symm-eq}
\end{corollary}
\begin{proof}
This corollary has two additional hypotheses beyond those in Proposition~\ref{pro:symm-eq}, i.e., $M=N+1$ and $\mathbf{A}$ has full column rank, and states two different conclusions from Proposition~\ref{pro:symm-eq}.
First, the necessary and sufficient condition for $\mathbf{A}\mathbf{X}=\mathbf{C}$ to have a symmetric solution is just $\mathbf{A}\mathbf{C}^T=\mathbf{C}\mathbf{A}^T$ instead of $\mathbf{A}\mathbf{C}^T=\mathbf{C}\mathbf{A}^T$ and $\mathbf{u}_2^T\mathbf{C}=\mathbf{0}_{1\times N}$.
Second, if $\mathbf{A}\mathbf{X}=\mathbf{C}$ is solvable, it has only one symmetric solution instead of the infinitely many given by \eqref{eq:symm-sol1}.
Thus, to prove this corollary, it suffices to prove that \textit{i)} $\mathbf{A}\mathbf{C}^T=\mathbf{C}\mathbf{A}^T$ implies $\mathbf{u}_2^T\mathbf{C}=\mathbf{0}_{1\times N}$ if $M=N+1$ and $\mathbf{A}$ has full column rank, and \textit{ii)} \eqref{eq:symm-sol1} boils down to \eqref{eq:symm-sol2} if $M=N+1$ and $\mathbf{A}$ has full column rank.

\textit{i)} Since $\mathbf{u}_2$ is the $(N+1)$th left singular vector of $\mathbf{A}$, which has rank $N$, we have $\mathbf{u}_2^T\mathbf{A}=\mathbf{0}_{1\times N}$.
Thus, $\mathbf{A}\mathbf{C}^T=\mathbf{C}\mathbf{A}^T$ implies $\mathbf{u}_2^T\mathbf{C}\mathbf{A}^T=\mathbf{0}_{1\times (N+1)}$.
In turn, $\mathbf{u}_2^T\mathbf{C}\mathbf{A}^T=\mathbf{0}_{1\times (N+1)}$ implies $\mathbf{u}_2^T\mathbf{C}=\mathbf{0}_{1\times N}$ since $\mathbf{A}$ has full column rank (recall that for any matrix $\mathbf{A}$ having full column rank, $\mathbf{k}\mathbf{A}^T=\mathbf{0}$ implies $\mathbf{k}=\mathbf{0}$).

\textit{ii)} It is easy to observe that \eqref{eq:symm-sol1} boils down to \eqref{eq:symm-sol2} if $\mathbf{A}$ has full column rank, i.e., $\mathbf{V}_1=\mathbf{V}$ and $\mathbf{V}_2$ is empty.
\end{proof}

Corollary~\ref{cor:symm-eq} can be applied to solve \eqref{eq:B22-cond-1-2} since $\mathbf{J}_1$ has dimensionality $N_S\times (N_S-1)$ and is full column rank with probability 1 given random channel realizations $\mathbf{H}$.
Therefore, \eqref{eq:B22-cond-1-2} has a symmetric solution $\mathbf{B}_{22,11}$ if and only if
\begin{equation}
\mathbf{J}_1\left(Y_0\mathbf{R}_1-\mathbf{J}_2\mathbf{B}_{22,21}\right)^T=\left(Y_0\mathbf{R}_1-\mathbf{J}_2\mathbf{B}_{22,21}\right)\mathbf{J}_1^T,\label{eq:cond-cor-tx}
\end{equation}
which can be shown as follows
\begin{align}
\mathbf{J}_1&\left(Y_0\mathbf{R}_1-\mathbf{J}_2\mathbf{B}_{22,21}\right)^T\\
&=Y_0\mathbf{J}_1\mathbf{R}_1^T-\mathbf{J}_1\mathbf{B}_{22,12}\mathbf{J}_2^T\\
&\stackrel{a}{=}Y_0\mathbf{R}_1\mathbf{J}_1^T-Y_0\mathbf{J}_2\mathbf{R}_2^T+Y_0\mathbf{R}_2\mathbf{J}_2^T-\mathbf{J}_1\mathbf{B}_{22,12}\mathbf{J}_2^T\\
&\stackrel{b}{=}Y_0\mathbf{R}_1\mathbf{J}_1^T-Y_0\mathbf{J}_2\mathbf{R}_2^T+\mathbf{J}_2\mathbf{B}_{22,22}\mathbf{J}_2^T\\
&=Y_0\mathbf{R}_1\mathbf{J}_1^T-\mathbf{J}_2\left(Y_0\mathbf{R}_2^T-\mathbf{B}_{22,22}\mathbf{J}_2^T\right)\\
&\stackrel{c}{=}Y_0\mathbf{R}_1\mathbf{J}_1^T-\mathbf{J}_2\mathbf{B}_{22,21}\mathbf{J}_1^T\\
&=\left(Y_0\mathbf{R}_1-\mathbf{J}_2\mathbf{B}_{22,21}\right)\mathbf{J}_1^T,
\end{align}
where
$\stackrel{a}{=}$ holds since $\mathbf{R}_1\mathbf{J}_1^T+\mathbf{R}_2\mathbf{J}_2^T=\mathbf{J}_1\mathbf{R}_1^T+\mathbf{J}_2\mathbf{R}_2^T$, which follows from $\mathbf{R}\mathbf{J}^T=\mathbf{J}\mathbf{R}^T$, verified since $\Re\{\bar{\mathbf{V}}\}^T\Im\{\bar{\mathbf{V}}\}=\Im\{\bar{\mathbf{V}}\}^T\Re\{\bar{\mathbf{V}}\}$ for any unitary matrix $\mathbf{V}=[\bar{\mathbf{V}},\tilde{\mathbf{V}}]$,
$\stackrel{b}{=}$ exploits $\mathbf{J}_2\mathbf{B}_{22,22}\mathbf{J}_2^T=Y_0\mathbf{R}_2\mathbf{J}_2^T-\mathbf{J}_1\mathbf{B}_{22,12}\mathbf{J}_2^T$ which follows from \eqref{eq:B22-cond-2},
$\stackrel{c}{=}$ exploits $(\mathbf{J}_1\mathbf{B}_{22,12})^T=(Y_0\mathbf{R}_2-\mathbf{J}_2\mathbf{B}_{22,22})^T$ following from \eqref{eq:B22-cond-2},
and the other equalities are straightforward.
Since we have verified \eqref{eq:cond-cor-tx}, Corollary~\ref{cor:symm-eq} states that \eqref{eq:B22-cond-1-2} has a unique symmetric solution as a function of the \gls{svd} of $\mathbf{J}_1$ and the matrix $Y_0\mathbf{R}_1-\mathbf{J}_2\mathbf{B}_{22,21}$.
Let the \gls{svd} of $\mathbf{J}_1$ be
\begin{equation}
\mathbf{J}_1=
\mathbf{U}_{J}
\begin{bmatrix}
\mathbf{\Sigma}_{J}\\
\mathbf{0}_{1\times (N_S-1)}
\end{bmatrix}
\mathbf{V}_{J}^T,\label{eq:svd-J1}
\end{equation}
where the matrix of left singular vectors $\mathbf{U}_{J}\in\mathbb{R}^{N_S\times N_S}$ is partitioned into $\mathbf{U}_{J}=[\mathbf{U}_{J,1},\mathbf{U}_{J,2}]$ with $\mathbf{U}_{J,1}\in\mathbb{R}^{N_S\times (N_S-1)}$ and $\mathbf{u}_{J,2}\in\mathbb{R}^{N_S\times 1}$, $\mathbf{\Sigma}_{J}\in\mathbb{R}^{(N_S-1)\times (N_S-1)}$ is a diagonal matrix containing the (positive) singular values, and $\mathbf{V}_{J}\in\mathbb{R}^{(N_S-1)\times (N_S-1)}$ is the matrix of right singular vectors.
Thus, $\mathbf{B}_{22,11}$ is given by
\begin{equation}
\mathbf{B}_{22,11}= \mathbf{V}_{J}\mathbf{\Sigma}_{J}^{-1}\mathbf{U}_{J,1}^T\left(Y_0\mathbf{R}_1-\mathbf{J}_2\mathbf{B}_{22,21}\right),
\end{equation}
according to Corollary~\ref{cor:symm-eq}.

We can now assemble the block $\mathbf{B}_{22}$ following \eqref{eq:B22-part} since we have derived the blocks $\mathbf{B}_{22,12}$, $\mathbf{B}_{22,21}$, and $\mathbf{B}_{22,22}$ in Section~\ref{sec:B22-blocks}, and the block $\mathbf{B}_{22,11}$ in Section~\ref{sec:B22-11}.

\subsection{Blocks $\mathbf{B}_{11}$, $\mathbf{B}_{12}$, and $\mathbf{B}_{21}$}

In Section~\ref{sec:B22}, we have derived in closed form the optimal value of the admittance matrix block $\mathbf{B}_{22}$.
With $\mathbf{B}_{22}$ known, the remaining blocks of the admittance matrix $\mathbf{B}_{12}$, $\mathbf{B}_{21}$, and $\mathbf{B}_{11}$, can be readily obtained as follows.
First, $\mathbf{B}_{12}$ can be found as a function of $\mathbf{B}_{22}$ by using \eqref{eq:tx-case4}, i.e.,
\begin{equation}
\mathbf{B}_{12}=-Y_0\mathbf{J}-\mathbf{R}\mathbf{B}_{22}.\label{eq:tx-B12}
\end{equation}
Second, recalling that $\mathbf{B}$ is symmetric, we have
\begin{equation}
\mathbf{B}_{21}=\mathbf{B}_{12}^T,\label{eq:tx-B21}
\end{equation}
with $\mathbf{B}_{12}$ set as in \eqref{eq:tx-B12}.
Third, the remaining block $\mathbf{B}_{11}$ can be found through \eqref{eq:tx-case3} as
\begin{equation}
\mathbf{B}_{11}=-\mathbf{R}\mathbf{B}_{21},\label{eq:tx-B11}
\end{equation}
with $\mathbf{B}_{21}$ given by \eqref{eq:tx-B21}.

We can now assemble $\mathbf{B}$ following \eqref{eq:B-part} since we have derived all its four blocks.
In Alg.~\ref{alg:tx}, we summarize the algorithm that allows us to find in closed form a capacity-achieving solution to the admittance matrix $\mathbf{B}$ of a stem-connected \gls{milac} at the transmitter.
From Alg.~\ref{alg:tx}, we notice that the computational complexity required to reconfigure a stem-connected \gls{milac} at each channel coherence time is driven by the complexity of computing the \gls{svd} of $\mathbf{H}$, i.e., $\mathcal{O}(N_T^2N_R)$ if $N_T\geq N_R$ or vice versa, same as the complexity to reconfigure a fully-connected \gls{milac} \cite{ner25-3}.
Besides, \gls{milac}-aided beamforming does not require any operation at each symbol time.

\subsection{Checking Solution Optimality}
\label{sec:tx-check}

Note that the solution in Alg.~\ref{alg:tx} is guaranteed by design to fulfill all the constraints of the feasibility check problem \eqref{eq:tx-prob}-\eqref{eq:tx-case4} except for two constraints that have not been used throughout our solution, namely \textit{i)} $\mathbf{B}_{11}=\mathbf{B}_{11}^T$ and \textit{ii)} $\mathbf{J}\mathbf{B}_{21}=-Y_0\mathbf{I}_{N_S}$.
For completeness, we show that these two constraints are satisfied in the following.
First, considering the block $\mathbf{B}_{11}$ set as in \eqref{eq:tx-B11}, we have
\begin{align}
\mathbf{B}_{11}^T
=-\mathbf{B}_{21}^T\mathbf{R}^T
&\stackrel{a}{=}Y_0\mathbf{J}\mathbf{R}^T+\mathbf{R}\mathbf{B}_{22}\mathbf{R}^T\\
&\stackrel{b}{=}Y_0\mathbf{R}\mathbf{J}^T+\mathbf{R}\mathbf{B}_{22}\mathbf{R}^T\\
&=\mathbf{R}\left(Y_0\mathbf{J}^T+\mathbf{B}_{22}\mathbf{R}^T\right)\\
&\stackrel{c}{=}-\mathbf{R}\mathbf{B}_{21}=\mathbf{B}_{11},
\end{align}
where
$\stackrel{a}{=}$ follows from \eqref{eq:tx-B12} and \eqref{eq:tx-B21},
$\stackrel{b}{=}$ follows from $\mathbf{R}\mathbf{J}^T=\mathbf{J}\mathbf{R}^T$, verified since $\Re\{\bar{\mathbf{V}}\}^T\Im\{\bar{\mathbf{V}}\}=\Im\{\bar{\mathbf{V}}\}^T\Re\{\bar{\mathbf{V}}\}$ for any unitary matrix $\mathbf{V}=[\bar{\mathbf{V}},\tilde{\mathbf{V}}]$, and
$\stackrel{c}{=}$ follows from $\mathbf{B}_{22}=\mathbf{B}_{22}^T$, \eqref{eq:tx-B12}, and \eqref{eq:tx-B21}, confirming that $\mathbf{B}_{11}=\mathbf{B}_{11}^T$.
Second, considering the block $\mathbf{B}_{21}$ set as in \eqref{eq:tx-B21}, we have
\begin{align}
\mathbf{J}\mathbf{B}_{21}
=\mathbf{J}\mathbf{B}_{12}^T
&\stackrel{a}{=}-\mathbf{J}\left(Y_0\mathbf{J}+\mathbf{R}\mathbf{B}_{22}\right)^T\\
&\stackrel{b}{=}-Y_0\mathbf{J}\mathbf{J}^T-\mathbf{J}\mathbf{B}_{22}\mathbf{R}^T\\
&\stackrel{c}{=}-Y_0\mathbf{J}\mathbf{J}^T-Y_0\mathbf{R}\mathbf{R}^T\stackrel{d}{=}-Y_0\mathbf{I}_{N_S},
\end{align}
where
$\stackrel{a}{=}$ follows from \eqref{eq:tx-B12},
$\stackrel{b}{=}$ follows from $\mathbf{B}_{22}=\mathbf{B}_{22}^T$,
$\stackrel{c}{=}$ follows from \eqref{eq:tx-case2}, and
$\stackrel{d}{=}$ follows from $\mathbf{J}\mathbf{J}^T+\mathbf{R}\mathbf{R}^T=\mathbf{I}_{N_S}$, verified since $\Im\{\bar{\mathbf{V}}\}^T\Im\{\bar{\mathbf{V}}\}+\Re\{\bar{\mathbf{V}}\}^T\Re\{\bar{\mathbf{V}}\}=\mathbf{I}$ for any unitary matrix $\mathbf{V}=[\bar{\mathbf{V}},\tilde{\mathbf{V}}]$, confirming that $\mathbf{J}\mathbf{B}_{21}=-Y_0\mathbf{I}_{N_S}$.

\section{Optimization of Capacity-Achieving Receiver-Side MiLAC Architectures}
\label{sec:rx}

In Section~\ref{sec:tx}, we have proven Proposition~\ref{pro:tx} by deriving a capacity-achieving solution to optimize stem-connected \glspl{milac} at the transmitter.
In this section, we prove Proposition~\ref{pro:rx} by extending the discussion of Section~\ref{sec:tx} to stem-connected \glspl{milac} at the receiver.
Although the optimization of receiver-side \glspl{milac} is analogous to that of transmitter-side \glspl{milac}, it requires a different partition of the \gls{milac} admittance matrix and hence deserves separate consideration.
To simplify the notation, only in this section, we denote the scattering, admittance, and susceptance matrices of the considered \gls{milac} at the receiver as $\boldsymbol{\Theta}$, $\mathbf{Y}$, and $\mathbf{B}$ instead of $\boldsymbol{\Theta}_G$, $\mathbf{Y}_G$, and $\mathbf{B}_G$.

\begin{figure}[t]
\centering
\includegraphics[height=0.44\textwidth]{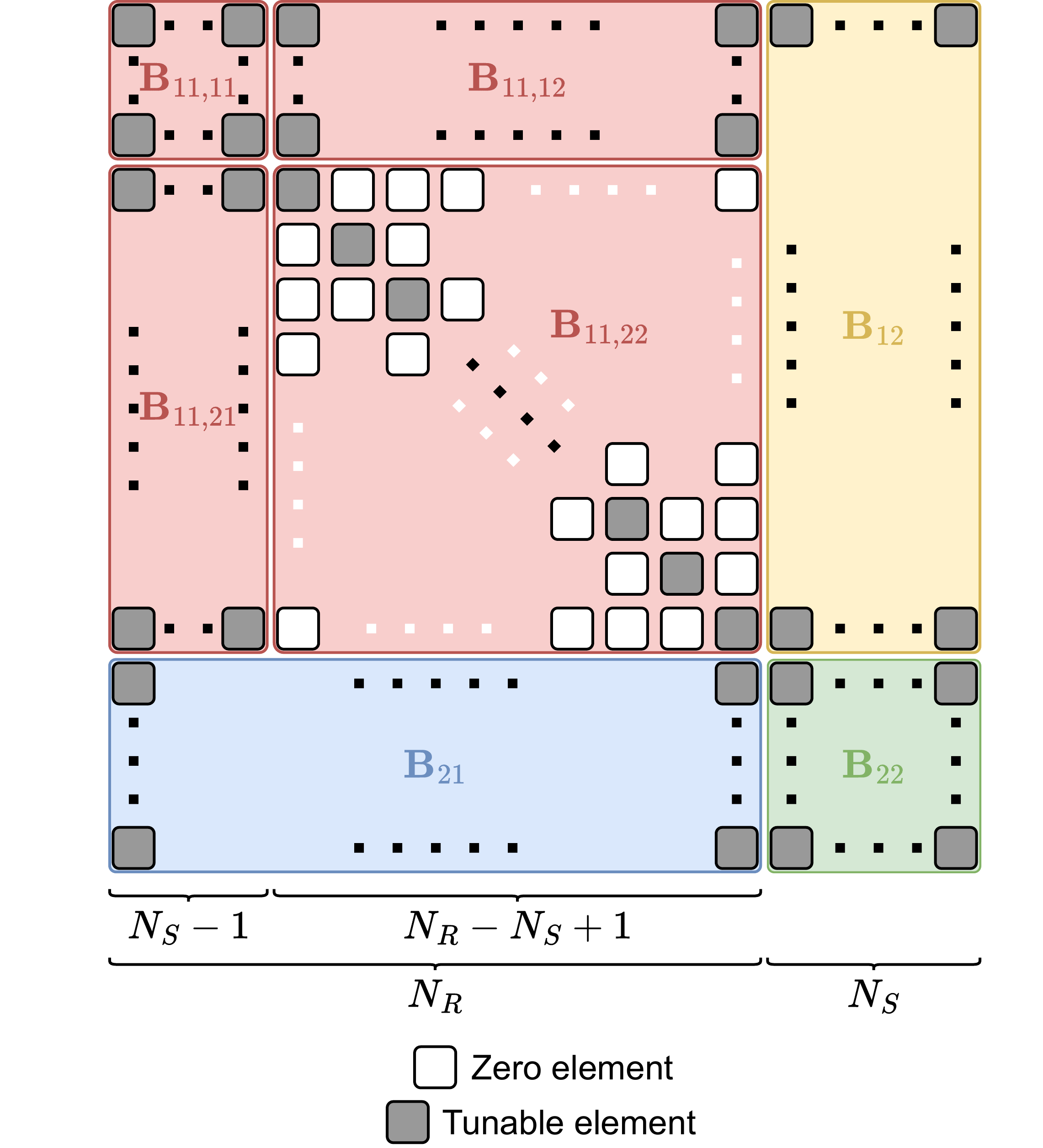}
\caption{Admittance matrix $\mathbf{B}$ of a stem-connected MiLAC at the
receiver partitioned into seven matrices, where $\mathbf{B}_{11,22}$ is diagonal.}
\label{fig:B-rx}
\end{figure}

\subsection{Reformulation of the Capacity-Achieving Condition}

Recalling that $\boldsymbol{\Theta}$ can be expressed as a function of the admittance matrix $\mathbf{Y}=j\mathbf{B}$ as in \eqref{eq:TF(YF)}, the capacity-achieving condition in \eqref{eq:rx-cond1} can be rewritten as
\begin{equation}
\left(Y_0\mathbf{I}-j\mathbf{B}\right)
\begin{bmatrix}
\mathbf{0}_{N_T\times N_S}\\
\mathbf{I}_{N_S}
\end{bmatrix}
=
\left(Y_0\mathbf{I}+j\mathbf{B}\right)
\begin{bmatrix}
\bar{\mathbf{U}}^*\\
\mathbf{0}_{N_S}
\end{bmatrix},
\end{equation}
as a function of $\mathbf{B}$, which can be simplified as
\begin{equation}
j\mathbf{B}
\begin{bmatrix}
\bar{\mathbf{U}}^*\\
\mathbf{I}_{N_S}
\end{bmatrix}
=Y_0
\begin{bmatrix}
-\bar{\mathbf{U}}^*\\
\mathbf{I}_{N_S}
\end{bmatrix},
\end{equation}
or, equivalently,
\begin{equation}
j
\begin{bmatrix}
\bar{\mathbf{U}}^H & \mathbf{I}_{N_S}
\end{bmatrix}
\mathbf{B}=Y_0
\begin{bmatrix}
-\bar{\mathbf{U}}^H & \mathbf{I}_{N_S}
\end{bmatrix},\label{eq:rx-cond3}
\end{equation}
given the symmetry of $\mathbf{B}$.
Since \eqref{eq:rx-cond3} is a linear matrix equation with real unknowns and complex coefficients, we rewrite it in real coefficients as
\begin{equation}
\begin{bmatrix}
\Re\left\{j
\begin{bmatrix}
\bar{\mathbf{U}}^H & \mathbf{I}_{N_S}
\end{bmatrix}\right\}\\
\Im\left\{j
\begin{bmatrix}
\bar{\mathbf{U}}^H & \mathbf{I}_{N_S}
\end{bmatrix}\right\}
\end{bmatrix}
\mathbf{B}=
\begin{bmatrix}
\Re\left\{Y_0
\begin{bmatrix}
-\bar{\mathbf{U}}^H & \mathbf{I}_{N_S}
\end{bmatrix}\right\}\\
\Im\left\{Y_0
\begin{bmatrix}
-\bar{\mathbf{U}}^H & \mathbf{I}_{N_S}
\end{bmatrix}\right\}
\end{bmatrix},
\end{equation}
which simplifies to
\begin{equation}
\begin{bmatrix}
\Im\left\{\bar{\mathbf{U}}\right\}^T & \mathbf{0}_{N_S}\\
\Re\left\{\bar{\mathbf{U}}\right\}^T & \mathbf{I}_{N_S}
\end{bmatrix}
\mathbf{B}=Y_0
\begin{bmatrix}
-\Re\left\{\bar{\mathbf{U}}\right\}^T & \mathbf{I}_{N_S}\\
\Im\left\{\bar{\mathbf{U}}\right\}^T & \mathbf{0}_{N_S}
\end{bmatrix}.\label{eq:rx-cond4}
\end{equation}
Similar to what done at the transmitter-side in Section~\ref{sec:tx}, we consider the partition of $\mathbf{B}$ in \eqref{eq:B-part}, where $\mathbf{B}_{11}\in\mathbb{R}^{N_R\times N_R}$, $\mathbf{B}_{12}\in\mathbb{R}^{N_R\times N_S}$, $\mathbf{B}_{21}\in\mathbb{R}^{N_S\times N_R}$, and $\mathbf{B}_{22}\in\mathbb{R}^{N_S\times N_S}$, and introduce $\mathbf{R}\in\mathbb{R}^{N_S\times N_R}$ and $\mathbf{J}\in\mathbb{R}^{N_S\times N_R}$ as $\mathbf{R}=\Re\left\{\bar{\mathbf{U}}\right\}^T$ and $\mathbf{J}=\Im\left\{\bar{\mathbf{U}}\right\}^T$, such that \eqref{eq:rx-cond4} becomes
\begin{equation}
\begin{bmatrix}
\mathbf{J} & \mathbf{0}_{N_S}\\
\mathbf{R} & \mathbf{I}_{N_S}
\end{bmatrix}
\begin{bmatrix}
\mathbf{B}_{11} & \mathbf{B}_{12}\\
\mathbf{B}_{21} & \mathbf{B}_{22}
\end{bmatrix}
=Y_0
\begin{bmatrix}
-\mathbf{R} & \mathbf{I}_{N_S}\\
\mathbf{J} & \mathbf{0}_{N_S}
\end{bmatrix}.\label{eq:rx-cond5}
\end{equation}
To formalize the constraints on the blocks $\mathbf{B}_{11}$, $\mathbf{B}_{12}$, $\mathbf{B}_{21}$, and $\mathbf{B}_{22}$, we assume with no loss of generality that the considered stem-connected \gls{milac} has the $N_S-1$ input ports $V_{1},\ldots,V_{N_S-1}$ connected to all the others, as Fig.~\ref{fig:center-rx}(a).
With this assumption, we have $\mathbf{B}\in\mathcal{B}_G$, where $\mathcal{B}_G$ is given by \eqref{eq:B-set-rx}.
Consequently, the only constraints on $\mathbf{B}_{12}$, $\mathbf{B}_{21}$, and $\mathbf{B}_{22}$ are $\mathbf{B}_{21}=\mathbf{B}_{12}^T$ and $\mathbf{B}_{22}=\mathbf{B}_{22}^T$ (due to the symmetry of $\mathbf{B}$, and since the $N_S$ output ports of the \gls{milac} are connected to all the others as required by Proposition~\ref{pro:rx}), while $\mathbf{B}_{11}$ satisfies $\mathbf{B}_{11}=\mathbf{B}_{11}^T$ and $\mathbf{B}_{11}\in\mathcal{B}_{11}$, where
\begin{multline}
\mathcal{B}_{11}=\left\{\mathbf{B}\in\mathbb{R}^{N_R\times N_R}\mid\left[\mathbf{B}\right]_{n,m}=0,\;\forall n\neq m,\right.\\
n>N_S-1,\;m>N_S-1\Bigl\}.\label{eq:B-set-11}
\end{multline}

In summary, the problem of finding a capacity-achieving solution for a stem-connected \gls{milac} at the receiver side is
\begin{align}
\mathrm{find}\;\;
&\mathbf{B}_{11},\;\mathbf{B}_{12},\;\mathbf{B}_{21},\;\mathbf{B}_{22}\label{eq:rx-prob}\\
\mathsf{\mathrm{s.t.}}\;\;\;
&\mathbf{B}_{11}=\mathbf{B}_{11}^T,\;\mathbf{B}_{21}=\mathbf{B}_{12}^T,\;\mathbf{B}_{22}=\mathbf{B}_{22}^T,\label{eq:rx-symm}\\
&\mathbf{B}_{11}\in\mathcal{B}_{11},\;\eqref{eq:B-set-11}\label{eq:rx-graph}\\
&\mathbf{J}\mathbf{B}_{11}=-Y_0\mathbf{R},\label{eq:rx-case1}\\
&\mathbf{J}\mathbf{B}_{12}=Y_0\mathbf{I}_{N_S},\label{eq:rx-case2}\\
&\mathbf{R}\mathbf{B}_{11}+\mathbf{B}_{21}=Y_0\mathbf{J},\label{eq:rx-case3}\\
&\mathbf{R}\mathbf{B}_{12}+\mathbf{B}_{22}=\mathbf{0}_{N_S},\label{eq:rx-case4}
\end{align}
where \eqref{eq:rx-symm} and \eqref{eq:rx-graph} follow from the constraints on $\mathbf{B}$, and \eqref{eq:rx-case1}-\eqref{eq:rx-case4} follow from the capacity-achieving condition \eqref{eq:rx-cond5}.

\begin{algorithm}[t]
\begin{algorithmic}[1]
\setstretch{1.25}
\REQUIRE $\bar{\mathbf{U}}$.
\ENSURE $\mathbf{B}$.
\STATEx{$\mathbf{R}=\left[\mathbf{R}_1,\mathbf{R}_2\right]=\Re\left\{\bar{\mathbf{U}}\right\}^T$,
$\mathbf{J}=\left[\mathbf{J}_1,\mathbf{J}_2\right]=\Im\left\{\bar{\mathbf{U}}\right\}^T$,}
\STATEx{$\mathbf{r}_i=[\mathbf{R}]_{:,i}$,
$\mathbf{j}_i=[\mathbf{J}]_{:,i}$,}
\STATEx{$\mathbf{J}_1=\left[\mathbf{U}_{J,1},\mathbf{u}_{J,2}\right]\left[\mathbf{\Sigma}_{J},\mathbf{0}_{(N_S-1)\times1}\right]^T\mathbf{V}_{J}^T$.}
\STATE{\textcolor{myred}{$\left[\mathbf{B}_{11,22}\right]_{i,i}=-Y_0
\left[\left[\mathbf{J}_1,\mathbf{j}_{N_S-1+i}\right]^{-1}\mathbf{r}_{N_S-1+i}\right]_{N_S}$.}}
\STATE{\textcolor{myred}{$\left[\mathbf{B}_{11,12}\right]_{j,i}=-Y_0
\left[\left[\mathbf{J}_1,\mathbf{j}_{N_S-1+i}\right]^{-1}\mathbf{r}_{N_S-1+i}\right]_{j}$.}}
\STATE{\textcolor{myred}{$\mathbf{B}_{11,21}=\mathbf{B}_{11,12}^T$.}}
\STATE{\textcolor{myred}{$\mathbf{B}_{11,11} = -\mathbf{V}_{J}\mathbf{\Sigma}_{J}^{-1}\mathbf{U}_{J,1}^T\left(Y_0\mathbf{R}_1+\mathbf{J}_2\mathbf{B}_{11,21}\right)$.}}
\STATEx{$\mathbf{B}_{11}=
\left[\left[\mathbf{B}_{22,11},\mathbf{B}_{22,12}\right]^T,
\left[\mathbf{B}_{22,21},\mathbf{B}_{22,22}\right]^T\right]^T$.}
\STATE{\textcolor{myblue}{$\mathbf{B}_{21}=Y_0\mathbf{J}-\mathbf{R}\mathbf{B}_{11}$.}}
\STATE{\textcolor{myyellow}{$\mathbf{B}_{12}=\mathbf{B}_{21}^T$.}}
\STATE{\textcolor{mygreen}{$\mathbf{B}_{22}=-\mathbf{R}\mathbf{B}_{12}$.}}
\STATEx{$\mathbf{B}=
\left[\left[\mathbf{B}_{11},\mathbf{B}_{12}\right]^T,
\left[\mathbf{B}_{21},\mathbf{B}_{22}\right]^T\right]^T$.}
\end{algorithmic}
\caption{Optimization of a stem-connected MiLAC at the receiver.}
\label{alg:rx}
\end{algorithm}

\subsection{Block $\mathbf{B}_{11}$}
\label{sec:B11}

We first determine $\mathbf{B}_{11}$ by solving \eqref{eq:rx-case1}.
To this end, we partition $\mathbf{B}_{11}$ as
\begin{equation}
\mathbf{B}_{11}=
\begin{bmatrix}
\mathbf{B}_{11,11} & \mathbf{B}_{11,12}\\
\mathbf{B}_{11,21} & \mathbf{B}_{11,22}
\end{bmatrix},\label{eq:B11-part}
\end{equation}
where $\mathbf{B}_{11,11}\in\mathbb{R}^{(N_S-1)\times (N_S-1)}$ is symmetric, $\mathbf{B}_{11,12}\in\mathbb{R}^{(N_S-1)\times (N_R-N_S+1)}$ and $\mathbf{B}_{11,21}\in\mathbb{R}^{(N_R-N_S+1)\times (N_S-1)}$ are such that $\mathbf{B}_{11,21}=\mathbf{B}_{11,12}^T$, and $\mathbf{B}_{11,22}\in\mathbb{R}^{(N_R-N_S+1)\times (N_R-N_S+1)}$ is diagonal since $\mathbf{B}_{11}\in\mathcal{B}_{11}$.
By substituting \eqref{eq:B11-part} into \eqref{eq:B-part}, we have that $\mathbf{B}$ is partitioned into a total of seven matrices, as represented in Fig.~\ref{fig:B-rx}.

By considering \eqref{eq:B11-part} and the previously introduced partitions $\mathbf{J}=[\mathbf{J}_1,\mathbf{J}_2]$ and $\mathbf{R}=[\mathbf{R}_1,\mathbf{R}_2]$, \eqref{eq:rx-case1} can be rewritten as
\begin{equation}
\begin{bmatrix}
\mathbf{J}_1 & \mathbf{J}_2
\end{bmatrix}
\begin{bmatrix}
\mathbf{B}_{11,11} & \mathbf{B}_{11,12}\\
\mathbf{B}_{11,21} & \mathbf{B}_{11,22}
\end{bmatrix}
=-Y_0
\begin{bmatrix}
\mathbf{R}_1 & \mathbf{R}_2
\end{bmatrix},
\end{equation}
which is equivalent to the system
\begin{numcases}{}
\mathbf{J}_1\mathbf{B}_{11,11} + \mathbf{J}_2\mathbf{B}_{11,21} = -Y_0\mathbf{R}_1,\label{eq:B11-cond-1}\\
\mathbf{J}_1\mathbf{B}_{11,12} + \mathbf{J}_2\mathbf{B}_{11,22} = -Y_0\mathbf{R}_2.\label{eq:B11-cond-2}
\end{numcases}
In the following, we derive solutions for the blocks of $\mathbf{B}_{11}$ fulfilling this system of matrix equations.
We first set $\mathbf{B}_{11,12}$, $\mathbf{B}_{11,21}$, and $\mathbf{B}_{11,22}$, and then optimize $\mathbf{B}_{11,11}$ accordingly.

\subsubsection{Blocks $\mathbf{B}_{11,12}$, $\mathbf{B}_{11,21}$, and $\mathbf{B}_{11,22}$}
\label{sec:B11-blocks}

Similar to what discussed in Section~\ref{sec:B22-blocks}, it is possible to show that \eqref{eq:B11-cond-2} is satisfied by setting the $i$th diagonal element of the diagonal matrix $\mathbf{B}_{11,22}$ as
\begin{equation}
\left[\mathbf{B}_{11,22}\right]_{i,i}=-Y_0
\left[\left[\mathbf{J}_1,\mathbf{j}_{N_S-1+i}\right]^{-1}\mathbf{r}_{N_S-1+i}\right]_{N_S},\label{eq:B1122ii}
\end{equation}
and the $(j,i)$th element of $\mathbf{B}_{11,12}$ as
\begin{equation}
\left[\mathbf{B}_{11,12}\right]_{j,i}=-Y_0
\left[\left[\mathbf{J}_1,\mathbf{j}_{N_S-1+i}\right]^{-1}\mathbf{r}_{N_S-1+i}\right]_{j},\label{eq:B1112ji}
\end{equation}
for $j=1,\ldots,N_S-1$ and $i=1,\ldots,N_R-N_S+1$, where $\mathbf{r}_i$ and $\mathbf{j}_i$ the $i$th column of $\mathbf{R}$ and $\mathbf{J}$, respectively, for $i=1,\ldots,N_R$.
According to $\mathbf{B}_{11,12}$ given by \eqref{eq:B1112ji}, $\mathbf{B}_{11,21}$ can be computed as $\mathbf{B}_{11,21}=\mathbf{B}_{11,12}^T$.

\subsubsection{Block $\mathbf{B}_{11,11}$}
\label{sec:B11-11}

We now need to find a symmetric matrix $\mathbf{B}_{11,11}$ such that \eqref{eq:B11-cond-1} is satisfied, i.e.,
\begin{equation}
\mathbf{J}_1\mathbf{B}_{11,11}=-Y_0\mathbf{R}_1-\mathbf{J}_2\mathbf{B}_{11,21},\label{eq:B11-cond-1-2}
\end{equation}
where $\mathbf{B}_{11,21}=\mathbf{B}_{11,12}^T$, with $\mathbf{B}_{11,12}$ set according to \eqref{eq:B1112ji}.
To solve the linear matrix equation in \eqref{eq:B11-cond-1-2}, we can apply Corollary~\ref{cor:symm-eq} as done in Section~\ref{sec:B22-11} since $\mathbf{J}_1$ has dimensionality $N_S\times (N_S-1)$ and is full column rank with probability 1 given random channel realizations $\mathbf{H}$.
Therefore, \eqref{eq:B11-cond-1-2} has a symmetric solution $\mathbf{B}_{11,11}$ if and only if
\begin{equation}
\mathbf{J}_1\left(Y_0\mathbf{R}_1+\mathbf{J}_2\mathbf{B}_{11,21}\right)^T=\left(Y_0\mathbf{R}_1+\mathbf{J}_2\mathbf{B}_{11,21}\right)\mathbf{J}_1^T,\label{eq:cond-cor-rx}
\end{equation}
which can be shown similarly as we have shown \eqref{eq:cond-cor-tx} in Section~\ref{sec:B22-11}.
Since \eqref{eq:cond-cor-rx} holds, Corollary~\ref{cor:symm-eq} states that \eqref{eq:B11-cond-1-2} has a unique symmetric solution given by
\begin{equation}
\mathbf{B}_{11,11}=-\mathbf{V}_{J}\mathbf{\Sigma}_{J}^{-1}\mathbf{U}_{J,1}^T\left(Y_0\mathbf{R}_1+\mathbf{J}_2\mathbf{B}_{11,21}\right),
\end{equation}
as a function of the \gls{svd} of $\mathbf{J}_1$ in \eqref{eq:svd-J1}.

Given the blocks $\mathbf{B}_{11,12}$, $\mathbf{B}_{11,21}$, and $\mathbf{B}_{11,22}$ derived in Section~\ref{sec:B11-blocks}, and the block $\mathbf{B}_{11,11}$ derived in Section~\ref{sec:B11-11}, we can now assemble $\mathbf{B}_{11}$ according to \eqref{eq:B11-part}.

\begin{figure}[t]
\centering
\subfigure[SNR $P_T/\sigma^2=0$~dB.]{
\includegraphics[height=0.30\textwidth]{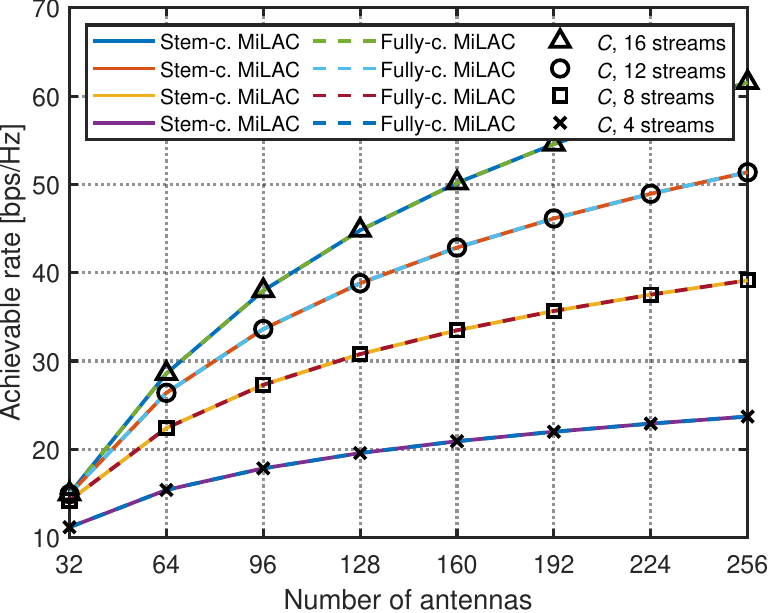}}
\subfigure[SNR $P_T/\sigma^2=10$~dB.]{
\includegraphics[height=0.30\textwidth]{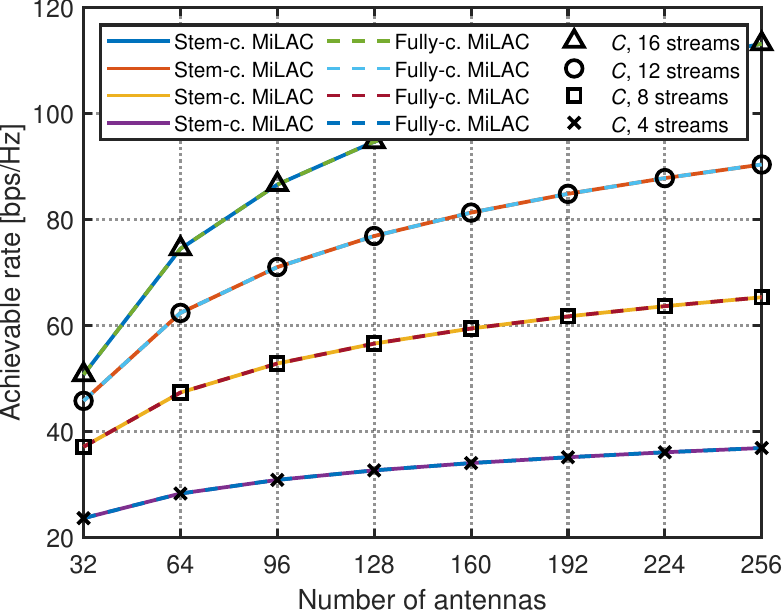}}
\caption{Achievable rate versus the number of antennas $N_T=N_R$, for different numbers of streams $N_S$, and values of SNR $P_T/\sigma^2$.
The achievable rate of stem- and fully-connected MiLAC, and the capacity in \eqref{eq:Cstar} are compared.}
\label{fig:capacity}
\end{figure}

\subsection{Blocks $\mathbf{B}_{12}$, $\mathbf{B}_{21}$, and $\mathbf{B}_{22}$}

With the block $\mathbf{B}_{11}$ derived in Section~\ref{sec:B11}, we can compute $\mathbf{B}_{21}$ by using \eqref{eq:rx-case3}, i.e.,
\begin{equation}
\mathbf{B}_{21}=Y_0\mathbf{J}-\mathbf{R}\mathbf{B}_{11},\label{eq:rx-B21}
\end{equation}
and, recalling that $\mathbf{B}$ is symmetric, we also have
\begin{equation}
\mathbf{B}_{12}=\mathbf{B}_{21}^T.\label{eq:rx-B12}
\end{equation}
Then, the remaining block $\mathbf{B}_{22}$ can be found through \eqref{eq:rx-case4} as
\begin{equation}
\mathbf{B}_{22}=-\mathbf{R}\mathbf{B}_{12},\label{eq:rx-B22}
\end{equation}
with $\mathbf{B}_{12}$ given by \eqref{eq:rx-B12}.

Finally, $\mathbf{B}$ is readily given by \eqref{eq:B-part} since we have derived all its four blocks.
In Alg.~\ref{alg:rx}, we summarize the steps leading to a capacity-achieving solution for the admittance matrix $\mathbf{B}$ of a stem-connected \gls{milac} at the receiver.
As discussed for a transmitter-side \gls{milac}, the computational complexity required to reconfigure a receiver-side \gls{milac} at each channel coherence time is driven by the computation of the \gls{svd} of $\mathbf{H}$, i.e., $\mathcal{O}(N_T^2N_R)$ if $N_T\geq N_R$ or vice versa.

\subsection{Checking Solution Optimality}
\label{sec:rx-check}

The solution in Alg.~\ref{alg:rx} is guaranteed by design to fulfill all the constraints of the feasibility check problem \eqref{eq:rx-prob}-\eqref{eq:rx-case4} except for two constraints that have not been used throughout our solution, namely \textit{i)} $\mathbf{B}_{22}=\mathbf{B}_{22}^T$ and \textit{ii)} $\mathbf{J}\mathbf{B}_{12}=Y_0\mathbf{I}_{N_S}$.
Interestingly, also these two constraints are satisfied, as it can be shown by proofs similar to those in Section~\ref{sec:tx-check}, which are omitted for conciseness.

\begin{figure}[t]
\centering
\includegraphics[height=0.30\textwidth]{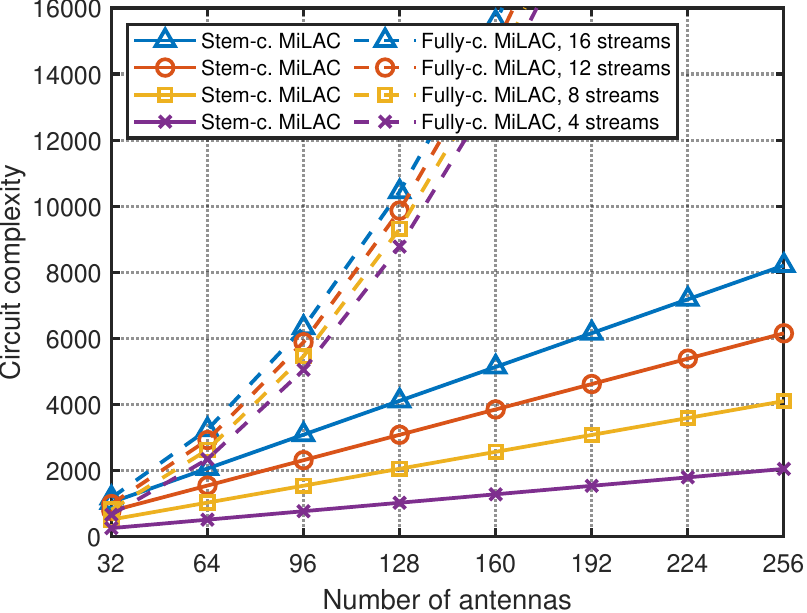}
\caption{Circuit complexity of stem- and fully-connected MiLAC versus the number of antennas ($N_T$ or $N_R$), for different numbers of streams $N_S$.}
\label{fig:complexity}
\end{figure}

\section{Numerical Results}
\label{sec:results}

We have proposed a novel family of \gls{milac} architectures, named stem-connected \gls{milac}, and provided a closed-form capacity-achieving solution to optimize it at the transmitter and receiver sides in Sections~\ref{sec:tx} and \ref{sec:rx}, respectively.
In this section, we validate the optimality of the proposed solution through Monte Carlo simulations and compare stem- and fully-connected \glspl{milac} in terms of performance and circuit complexity.

In Fig.~\ref{fig:capacity}, we report the achievable rate obtained with stem-connected and fully-connected \glspl{milac} averaged over sufficient realizations of \gls{iid} Rayleigh fading channels
\footnote{Although we consider Rayleigh fading channels, the optimization algorithms in Sections~\ref{sec:tx} and \ref{sec:rx} are proved to be globally optimal for any channel model, including more practical dual-polarized channels by properly adapting the number of \gls{milac} ports to the number of antennas.}
, i.e., $\text{vec}(\mathbf{H})\sim\mathcal{CN}(\mathbf{0}_{N_RN_T\times 1},\mathbf{I}_{N_RN_T\times 1})$.
We consider different scenarios where $N_R=N_T$, the number of streams is $N_S\in\{4,8,12,16\}$ and the \gls{snr} is $P_T/\sigma^2\in\{0,10\}$~dB.
For each scenario, we compare: the achievable rate with stem-connected \glspl{milac} at the transmitter and receiver sides optimized with the solutions proposed in Sections~\ref{sec:tx} and \ref{sec:rx}, the achievable rate with fully-connected \glspl{milac} optimized as proposed in \cite{ner25-3}, and the capacity in \eqref{eq:Cstar}.
In agreement with our theoretical analysis, we observe that stem-connected \glspl{milac} achieve the same rate as fully-connected \glspl{milac}, corresponding to the capacity derived in \eqref{eq:Cstar}.
Thus, while maintaining the maximum performance, stem-connected \glspl{milac} can provide enormous gains in terms of circuit complexity.

In Fig.~\ref{fig:complexity}, we report the circuit complexity, i.e., the number of admittance components, of a stem-connected \glspl{milac} $N_C^\text{Stem}$ with center size $Q=2N_S-1$ given by \eqref{eq:stem-tx-complexity}, and of a fully-connected \gls{milac} $N_C^\text{Fully}$ given by \eqref{eq:fully-tx-complexity}.
The circuit complexity is reported versus the number of antennas ($N_T$ for a \gls{milac} at the transmitter or $N_R$ for a \gls{milac} at the receiver), for different values of the number of streams $N_S$.
We make the following three observations.
\textit{First}, the circuit complexity of both fully- and stem-connected \glspl{milac} increases with the number of streams and the number of antennas.
\textit{Second}, with a fixed number of streams $N_S$, the circuit complexity of a fully-connected \gls{milac} scales quadratically with the number of antennas, i.e., $N_C^\text{Fully}=\mathcal{O}((N_S+N_T)^2)$ or $N_C^\text{Fully}=\mathcal{O}((N_S+N_R)^2)$ for a \gls{milac} at the transmitter or at the receiver, respectively.
\textit{Third}, with a fixed $N_S$, the circuit complexity of a stem-connected \gls{milac} scales linearly with the number of antennas, i.e., $N_C^\text{Stem}=\mathcal{O}(N_SN_T)$ or $N_C^\text{Stem}=\mathcal{O}(N_SN_R)$ for a \gls{milac} at the transmitter or at the receiver, respectively.
Thus, a stem-connected \gls{milac} can achieve the same capacity as a fully-connected \gls{milac}, simplifying significantly the hardware when the number of antennas grows large.
This gain in terms of circuit complexity of stem- over fully-connected \gls{milac} is expected to also reflect the gain in power consumption.
This is because the power consumption of a \gls{milac} is due to the control voltages used to reconfigure the tunable admittance components, and therefore scales with the number of tunable admittance components.

\section{Conclusion}
\label{sec:conclusion}

The concept of \gls{milac} has recently emerged to enable gigantic \gls{mimo} by performing signal processing and beamforming entirely in the analog domain via multiport microwave networks.
However, previous work has explored \gls{milac} architectures, denoted as fully-connected \glspl{milac}, which include tunable impedance components interconnecting all \gls{milac} ports to each other.
Hence, their circuit complexity, measured in terms of the number of required impedance components, scales quadratically with the number of antennas, limiting their practicality.
In this paper, we address this scalability issue by proposing a graph theoretical model of \gls{milac} that allows us to explore lower-complexity \gls{milac} architectures.
We model a \gls{milac} as a graph, whose vertices represent the \gls{milac} ports and edges represent the presence of impedance components interconnecting the ports.
Through this graph theoretical model, we characterize a family of \gls{milac} architectures, denoted as stem-connected \glspl{milac}, which maintains the capacity-achieving property of fully-connected \glspl{milac}, while significantly reducing the circuit complexity.
To show that stem-connected \glspl{milac} are capacity-achieving, we optimize them through a closed-form solution that is proven to globally maximize the achievable rate for any channel realization, i.e., to reach the capacity.
Simulation results support our theoretical analysis and demonstrate that stem-connected \glspl{milac} can achieve the same capacity as fully-connected \glspl{milac}, which corresponds to the capacity achievable with a fully-digital \gls{mimo} system with the same number of streams.
At the same time, stem-connected \glspl{milac} are characterized by a circuit complexity scaling linearly with the number of antennas, rather than quadratically, making them a practical and scalable solution to enable high-performance gigantic \gls{mimo} systems.

In this study, we have proposed low-complexity \gls{milac} architectures that are proven to achieve capacity in point-to-point \gls{mimo} systems, and we have developed an algorithm to globally optimize them.
We identify several interesting future research directions, both theoretical and practical.
A promising theoretical direction is to investigate whether the proposed \gls{milac} architectures remain optimal in multi-user communication systems, and to develop optimization strategies for such scenarios.
Since we have proved that $Q=2N_S-1$ is a sufficient condition for a stem-connected \gls{milac} to be capacity achieving, it would be also interesting to prove that this condition is also necessary, or alternatively, to design capacity-achieving architectures with even lower circuit complexity.
Additional directions could include the development of more realistic models for \gls{milac}, e.g., accounting for lossy \gls{rf} components and for components whose values are restricted to a limited range.
The optimal \gls{milac} architectures and their optimization under such non-ideal conditions could be investigated.
Finally, future research should explore the prototyping of a small-scale, stem-connected \gls{milac}, to demonstrate its practical benefits.
Implementation challenges include:
\textit{i)} ensuring an accurate tuning of the admittance components to the target values,
\textit{ii)} mitigating parasitic effects and mutual coupling among closely spaced interconnections, and
\textit{iii)} preserving component linearity and stability across different frequencies, particularly important in wideband communications.

\bibliographystyle{IEEEtran}
\bibliography{IEEEabrv,main}

\end{document}